\documentclass[11pt,letter]{article}

\usepackage[margin=2cm]{geometry}

\usepackage[ruled]{algorithm2e}
\usepackage[hidelinks]{hyperref}

\usepackage[utf8]{inputenc}

\usepackage{subfig}
\usepackage[normalem]{ulem}
\usepackage{float}
\usepackage{fancyhdr}
\usepackage{amssymb}
\usepackage{algpseudocode}
\usepackage{amsfonts}
\usepackage{graphicx}
\usepackage{tikz}
\usepackage{hyperref}
\usepackage{outlines}
\usepackage{mathtools}

\newcommand{\alg}[1]{\textnormal{\scshape #1}}

\newcommand{\ignore}[1]{}

\newtheorem{theorem}{Theorem}

\newenvironment{proof}{\noindent\bf{Proof }\rm}{\hfill$\blacksquare$\bigskip}

\newcommand{\eps}{\epsilon}
\newcommand{\OGT}{\mathsf{OnlineGreedyTrackingV2}}

\newcommand{\px}{p_{max}}

\newcommand{\Span}{\textsf{Span}} 
\newcommand{\ALG}{\operatorname{ALG}}
\newcommand{\OPT}{\textsf{OPT}} 
\newcommand{\opt}{\textsf{opt}} 
\newcommand{\hj}{\hat{j}} 
\newcommand{\bI}{\bar{I}} 
\newcommand{\tI}{\tilde{I}} 
\newcommand{\hI}{\hat{I}} 
\newcommand{\bs}{\bar{s}} 
\newcommand{\cost}{\text{cost}}

\newcommand{\QQ}{{\mathbb{Q}}}
\newcommand{\ZZ}{{\mathbb{Z}}}

\newcommand{\cA}{\mathcal{A}}

\usepackage{parskip}
\usepackage{setspace}

\begin{document}

\title{On Randomized Online Span Minimization}
\author{
Adrian Calinescu\thanks{Authors are listed in order of decreasing contribution.
{\tt adriancyay@gmail.com}.}
\and
Gruia C\u alinescu\thanks{Work completed in part while at Illinois Institute of Technology.
{\tt gcali9999@gmail.com}.}
\and
Peng-Jun Wan\thanks{pengjunwan@gmail.com}
}%

\date{}
\maketitle

\begin{abstract}
We study the online Busy Time scheduling model on a single machine of unbounded
capacity, with non-preemptive jobs. In our setting, ``flexible" jobs 
arrive online with a processing time and deadline, both of which become known to the algorithm at the job's arrival time.
The goal is to schedule jobs on the machine to finish all jobs by their deadlines,
so that the total time when the machine is turned on (busy time, also called {\em span} in this setting) is minimized.

We present a randomized online algorithm with an exact competitive ratio of 
$1 + e < 3.72$, where $e$ is Euler's number,
against an oblivious adversary,
and show that no randomized algorithm  can have a competitive ratio better than $e$
against an oblivious adversary.  This  lower bound
holds even for algorithms that are allowed to restart jobs and are given lookahead. 
Previous work either dealt with special cases,  or with deterministic algorithms,
for which the known upper bound on the competitive ratio is $5$
and the known lower bound is $4$.
Our findings offer fresh insights into randomization in online energy-aware scheduling.

In the setting where jobs have uniform processing times, and
a job that is started by the algorithm must be finished, we show that
no deterministic algorithm can do better than 2, even when the jobs are agreeable\footnote{Independently obtained in \cite{LKT26}.}.
This deterministic lower bound also holds for uniform processing times and 
restarts in the scenario where a job's deadline is
only revealed at its starting deadline.

For Capacitated Busy Time with  $p_{max}$-lookahead, we obtain
a deterministic online algorithm with competitive ratio at most $8$
and a randomized  online algorithm with competitive ratio at most $4+ e  < 6.72$,
against an oblivious adversary.
\end{abstract}

\section{Introduction}

We study the Busy Time Scheduling problem in the following setting.
Job $j$ is characterized by release time
$r_j$, processing time $p_j$, and deadline $d_j$.
The interval $[r_j,d_j)$ is also called the {\em window} $W_j$ of job $j$. 
Every job is feasible: $0<p_j\le d_j-r_j$.
Jobs are flexible in the sense that the size of the window of a job can be much larger than its processing time.
The algorithm must decide for each 
job $j$  an interval $I_j$ of length  $p_j$ contained in $W_j$ during which job $j$ is
to be processed by the machine. In this work, the machine can be turned on and off,
and has enough capacity to accommodate all the jobs at the same time, i.e., in parallel. 
The objective is to minimize the
so-called {\em busy machine time}, 
which is the  total length of $\bigcup_j I_j$,
also called {\em span} in the literature.
Formally, for a measurable set of real numbers $V$, we use 
$\lambda(V)$ to denote the Lebesgue measure of $V$,
and the objective is to minimize
$\lambda \left( \bigcup_j I_j \right)$.

This problem is motivated by applications in cloud computing.
Cloud computing is essential in today’s digital landscape due to its numerous benefits.
Energy is a crucial and limited resource, and its consumption is becoming an increasingly important concern.
Computing environments contribute significantly to global energy consumption, and this share is rapidly expanding \cite{cloudenergy}.
To address this, modern hardware is progressively integrating various energy-saving features.
Consequently, scheduling algorithms must be designed with a focus not only on time and space efficiency but also on minimizing energy consumption.
One of the most prevalent power-management techniques
is the power-down mechanism, which allows the processor to enter
a sleep state where it consumes minimal energy.
This gives rise to Busy Time problems.

The busy-time metric is also closely related to several important problems in optical network design.
For instance, it plays a key role in reducing the fiber costs associated with Optical Add-Drop Multiplexers (OADMs) \cite{FMMSSTZ10}.
The use of busy-time models in optical network design has been widely studied and documented in the literature
\cite{FMMSSTZ10,FMMSZ21,winkler2003wavelength,alicherry2003line}.

In this work, 
we are primarily interested in the online setting (see \cite{BEY98} for an introduction to online algorithms),
where a job $j$ becomes known to the algorithm at
its release time $r_j$.

For an instance (set of jobs) $J$, 
define $\opt(J)$ to be the objective of an optimum solution on input $J$.
For an algorithm $\ALG$, define $\ALG(J)$ to be the cost of $ALG$
on input $J$. The competitive ratio of $A$ is:
\[
\sup_{J \, \OPT(J) > 0} \frac{\ALG(J)}{\OPT(J)}.
\]
Above we assumed that $\ALG$ is a deterministic algorithm, and 
therefore we have an {\em adversary} who chooses the instance $J$ knowing
exactly what $\ALG$ will do on input $J$.
If $\ALG$ is a randomized algorithm, its competitive ratio is:
\[
\sup_{J \, \OPT(J) > 0}  \frac{\mathbb E[\ALG(J)]}{\OPT(J)}.
\]
The expectation is over the internal randomness of $\ALG$.
When it comes to randomized online algorithms, this paper only uses
{\em oblivious} adversaries.
An oblivious adversary \cite{BEY98} knows the randomized algorithm but must 
choose the instance before the algorithm starts running
(and so, the oblivious adversary cannot
adapt based on the algorithm’s random choices). 
The paper uses an absolute multiplicative ratio with no additive constant\footnote{
The lower bound examples can be repeated so that the theorems hold
with an additive constant as well.
}.
We say that an algorithm is $c$-competitive if its competitive ratio is at most $c$.

Coming back to Span Minimization,
in the offline setting, where all the jobs are known in advance, minimizing the
span has a dynamic-programming-based polynomial-time exact algorithm \cite{KSST15} 
(the paper also considers machines of bounded capacity). 
Deterministic online algorithms appeared in  \cite{ren2017online} and \cite{koehler2017busy},
with the ``Doubler" algorithm of \cite{koehler2017busy} achieving the best known competitive ratio of $5$.
References \cite{Fang2013ApplicationAware,Fong2015ActiveTime,Fong2017ActiveTime} claim better  competitive ratios but there
are doubts about their correctness - see \cite{LT24} and the v1 version of \cite{albers2025onlinebusytimescheduling}.
As for lower bounds on the competitive ratio,
Liu and Tang \cite{LT24} (see also the v1 version of \cite{albers2025onlinebusytimescheduling}) claim the best known bound, which is $4$.

We are the first to study randomized online algorithms for this problem with arbitrary processing times. 
Our randomized algorithm achieves a competitive
ratio of $(1+e) < 3.72$ and uses a randomized scaling method inspired by \cite{SLLA25} 
(the earliest paper  that we know of with this method is \cite{CNS04}) and the
analysis of our algorithm uses ideas from \cite{koehler2017busy}.
Note that this randomized competitive  ratio is better than
any deterministic algorithm can achieve.

\begin{theorem} \label{thm: randomized}
In the oblivious adversary model,
there is a $(1 + e )$-competitive randomized  online algorithm for Span Minimization,
with running time $O(n \log n)$, where $n$ is the total number of jobs.
\end{theorem}

We note that both this algorithm and the previously known $5$-competitive deterministic algorithm of \cite{koehler2017busy}
have an easy implementation to run in $O(n \log n)$ time (details in Section \ref{s_alg}).
As a result, these online algorithms have some value
as approximation algorithms, as they are much faster than the dynamic program of \cite{KSST15}\footnote{
We also have code from ChatGPT for both our algorithm
and the exact algorithm of \cite{KSST15}. 
The C++ code is available at https://github.com/gcal9999/Span-Minimization.git. 
Within 30 seconds, the randomized online algorithm can solve random instances with ten million jobs,
versus two thousand jobs for the exact algorithm.
The approximate solution is $18\text{--}33\%$ bigger 
(we only checked ten instances; this is not a simulation study).
}
. In Appendix \ref{a_tight} we show that the $1+e$ bound is tight for our 
randomized algorithm.

Many research works deal with Busy Time on bounded capacity machines, as we mention later.
Below, $g$ is the maximum number of jobs that can be scheduled at the same time on a machine.
We assume an unlimited number of identical machines.
The algorithm must decide for each 
job $j$  an interval $I_j$ of length  $p_j$ contained in $W_j$ during which job $j$
is to be processed, and assign (schedule) the job to one machine.
At any time, a machine  is {\em busy} if there is at least one job scheduled on it.
The objective of Busy Time scheduling is to minimize the total busy time over all time and machines.

Shalom et al. \cite{DBLP:journals/tcs/ShalomVWYZ14} proved that no deterministic online
algorithm can achieve a competitive ratio better than $g$  for 
Busy Time in this setting.
In fact, this lower bound holds even if there is no flexibility 
in the {\em starting time} (that is, $d_j-p_j=r_j$)
for every job $j$; such jobs are called {\em rigid}.
One way to mitigate this strong lower bound is to allow {\em lookahead}.
With a lookahead of $\ell$, job $j$ is revealed at time $r_j-\ell$.
Let $p_{max} := \max_{j}(p_j)$. 
As the lookahead will be a function of $\px$,
we assume that $\px$ is known to the algorithm from the start.
Koehler and Khuller \cite{koehler2017busy}
obtained a $12$-competitive algorithm with a lookahead of $2p_{max}$.
The v2 version of \cite{albers2025onlinebusytimescheduling}
claims a $9$-competitive online algorithm with only $p_{max}$-lookahead,
but their proof is incorrect,
as we show in Appendix \ref{a_albers} referenced from Section \ref{s_cap} (and if it can be corrected, we would get a randomized $(10/3 +e) < 6.06$-competitive algorithm).
However, a new algorithm inspired by their method is correct and,
with our better analysis, we obtain in Subsection \ref{ss_3tracks}:
\begin{theorem}
\label{t_cap}
    With $\px$-lookahead, there exists a deterministic 
    online algorithm for Capacitated Busy Time
    with competitive ratio of at most $8$, and
    there exists a randomized online algorithm for Capacitated Busy Time
    with competitive ratio of at most $4+e < 6.72$, against an oblivious adversary.
    When the jobs are rigid, the competitive ratio is $4$.
\end{theorem}

Also, using our algorithm from Theorem \ref{thm: randomized} and a variant of the \cite{koehler2017busy} algorithm
we obtain in Section \ref{s_cap} 
a $\left( 2 + (l+1)(e+1)\right)$-competitive algorithm with $\px/l$-lookahead, for any integer $l \geq 1$.
Proposition 5  of
\cite{koehler2017busy} provides a $\sqrt{2}$  lower bound on the competitive ratio of deterministic algorithms
with a $p_{max}$-lookahead. Theorem \ref{thm: lb_nested_deterministic}, below, improves this lower bound to 4 since the case $g=n$ is the same as minimizing span.

We also prove a lower bound of $e$ on the competitive ratio of any randomized algorithm versus an oblivious adversary.
We show that this lower bound holds even if aborting and restarting a job from scratch is allowed.
To be precise, in the restart model, aborted work is lost, aborted execution still
        contributes to span, and the job must later complete in one uninterrupted attempt.
We assume that an online algorithm can only do finitely many restarts.
\cite{LKT26} demonstrates that allowing restarting jobs does produce a better competitive ratio in the case of  uniform processing time jobs;
in general restarting can bring benefits in online scheduling \cite{AmouzandehS26}.
Our lower bound comes from a ``reduction" from the Button problem of
Shin et al. \cite{SLLA25}, 
who show that no randomized online algorithm can achieve a competitive ratio better than $e$ for this
very basic Button problem. 
We say that two jobs are {\em nested} if one job's window is a subset of the other job's window.
We say that two jobs are {\em w-disjoint} if their windows are disjoint.

\begin{theorem} \label{thm: lb_nested_randomized}
For any $\epsilon > 0$,
no randomized online algorithm can have competitive ratio $e - \epsilon$ for
Span Minimization, against an oblivious adversary.
This holds even when every two jobs are nested or w-disjoint\footnote{The ``nested or w-disjoint" instance class is also called {\em laminar}, for example in \cite{CL19}.},
 and even when the algorithm is given a lookahead that is a function of $p_{max}$ and is allowed to abort and restart jobs.
\end{theorem}

Basically the same reduction also gives 
the lower bound of $4$ previously obtained by \cite{LT24}, and we note that this lower bound  
 also holds even if restarting a job is allowed. Conceptually, there are
 advantages to using a reduction, as it separates the proof of the lower bound of $4$
 into two parts that can be used separately (and indeed, we do use one such part for the randomized lower bound).
 We also believe that using this reduction makes understanding the proofs of the lower bounds easier,
 so we decided  to present the following result:

 \begin{theorem} \label{thm: lb_nested_deterministic} [basically, Theorem  2.6 of \cite{LT24}]
For any $\epsilon > 0$,
no deterministic online algorithm can have competitive ratio $4 - \epsilon$ for Span Minimization.
This holds even when every two jobs are nested or w-disjoint,
 and even when the algorithm is given a lookahead that is a function of $p_{max}$ and is allowed to abort and restart jobs.
\end{theorem}

 This is not a new result; we just observe it has extra properties.
As for any reduction, we give a construction and related proofs.
 After we obtained our reduction, we noticed that,
 if used for the lower bound of $4$,
 our construction differs in some respects from 
 the proof we have seen in a personal communication
 of a full version of \cite{LT24},
 but it is ``essentially the same"  as 
 the construction suggested in the v1 version of \cite{albers2025onlinebusytimescheduling}
(their Section 2.2 is somewhat vague and lacks some arguments, notably how the adversary schedules all the jobs).
 Our proof also 
 underlines the usefulness in online scheduling of the Button problem of \cite{SLLA25}. 

A set of jobs is called {\em agreeable}\footnote{Also called {\em compatible},
for example in \cite{Lawler1991ADP}, or {\em proper}, for example in \cite{CL19}.
}
if, for any two jobs $j$ and $j'$ in the set, if
job $j$ arrives before job $j'$,
then the deadline of job $j'$ cannot be before that of $j$.

Agreeable instances can be easier, and indeed 
Albers and van der Heijden \cite{albers2025onlinebusytimescheduling} (v2) obtain a $2$-competitive
online deterministic algorithm  (also, their  Algorithm 2 does not use restarts)
for agreeable instances of Span Minimization,
which is optimal in light of the theorem below.
With agreeable jobs, faster (offline) exact algorithms
based on dynamic programming were given in \cite{Fong2017ActiveTime}.

In the setting where jobs have uniform processing times, and
a job that is started by the algorithm must be finished,
we show that no deterministic algorithm can have competitive ratio better than 2, even when the
jobs are agreeable.
We obtained and communicated this result in 2024, without input from  \cite{LKT26},  which also claims this.
\cite{albers2025onlinebusytimescheduling} (v2 only, from 2025) also claims this lower bound of $2$ in their Theorem 1,
but their method is vague and seems to have the adversary change the deadlines of some of the jobs
after the online algorithm has scheduled these jobs. Our proof is inspired by the proof of Theorem 3 of \cite{CDKZ24}.
It is straightforward to verify that Doubler is $2$-competitive for uniform processing times
(agreeable instance or not).

\begin{theorem} [Also Section 2.2 of \cite{LKT26}]
    \label{t_2}
    For any $\epsilon > 0$,
    no deterministic online algorithm can have competitive ratio $2 - \epsilon$ for
     Uniform Processing Time  Span Minimization with agreeable jobs, assuming restarts are not allowed.
\end{theorem}

This theorem improves on the $\frac{1 + \sqrt{5}}{2} \approx 1.618$ lower bound 
of \cite{koehler2017busy,ren2017online} and the v1 version of  \cite{albers2025onlinebusytimescheduling}.

The following online scenario also makes sense: at release time, we know a job's processing time, but not the deadline.
The job's deadline is only revealed at the job's {\em starting deadline} $d_j-p_j$
(when the job becomes rigid).
All algorithms without restarts and with established guarantees mentioned above
(ours or from previous work) work in this {\em flexible-rigid} scenario. The method of Theorem \ref{t_2} can be adapted to
prove that restarts do not help in the flexible-rigid model:
\begin{theorem}
    \label{t_fs}
    For any $\epsilon > 0$,
    no deterministic online algorithm can have competitive ratio $2 - \epsilon$ for
     Uniform Processing Time  Span Minimization with agreeable jobs
     in the flexible-rigid model, even if restarts are allowed.
\end{theorem}

Most related to this work is the recent work \cite{LKT26}. 
They demonstrate that randomization can be leveraged to break the deterministic competitive barrier of $2$ 
of Theorem \ref{t_2}. For uniform processing time jobs, they present
a randomized competitive upper bound of $\frac{1}{\ln{2}}\approx 1.443$
and a lower bound of $\frac{\sqrt{3}+1}{2}\approx 1.366$,
both against an oblivious adversary. 
In the deterministic setting, they show that by allowing job restarts,
one can achieve an optimal competitive ratio equal to $\frac{1 + \sqrt{5}}{2}$.
Their randomized algorithms without restarts remain valid in the flexible-rigid scenario,
while the algorithm with restarts does not
(which is expected in light of the above theorem).

\subsection{Related work}

When the jobs can be preempted, span can be minimized 
by the algorithm from Theorem 6 of \cite{CKM17}. It is straightforward
to adapt this algorithm to find an optimum also in the online setting.

In the setting of  bounded capacity machines, with an unlimited number of identical machines,
recall that $g$ is the maximum number of jobs that can be scheduled at the same time on a machine.
Chau and Li \cite{CL19} provide a survey. The problem becomes NP-hard \cite{winkler2003wavelength},
even if there is no flexibility for the jobs (i.e., for all $j$, $p_j = d_j - r_j$).
Approximation algorithms are presented in 
\cite{alicherry2003line,kumar2005approximation,FMMSSTZ10,KSST15,CKM17}.
Koehler and Khuller \cite{koehler2017busy} 
give an $O(\log \mu_p)$ approximation for general $g$,
where $\mu_p$ is the ratio of maximum to minimum processing time.
This implies a constant approximation ratio in the case of uniform processing times.
A constant ratio for uniform processing times is also obtained in \cite{CDKZ24}, even with heterogeneous machine types.
Also, with bounded capacity machines and inflexible jobs,
online algorithms have been published in \cite{azar2019tight,Kamali2015Efficient,Li2014Dynamic,RTLC17}.

With preemption and migration,   Chang et al.
\cite{CKM17} also obtain a $2$-approximation algorithm for bounded capacity Busy Time, 
building on their algorithm that minimizes the span. 
When $g$ is part of the input, this problem is also NP-hard as the reduction of \cite{CaoFLMRU22} works for Busy Time as well.
In this scenario, if we have $d$ types of resources and
$d$ capacity constraints (the same for all the machines),  \cite{sarpatwar2023preemptive} obtain an
$O(\log d \log^* \tau')$-approximation algorithm, where $\tau'$ is the total number of timeslots.

\section{The Randomized Online Algorithm}
\label{s_alg}

\begin{proof}[\emph{of Theorem \ref{thm: randomized}}]

The randomized algorithm chooses uniformly at random a number $y \in [0,1)$ (this is the only random choice).
Then it proceeds as follows.  Let $t$ be the current time.
First, process arrivals:
When job $h$ arrives at time $t$,
schedule it to start immediately into an already-activated interval if it
fits\footnote{formally, job $h$ fits in interval $I$ if $\lambda\bigl(I\cap W_h\bigr)\ge p_h$}.
Second, process starting deadlines:
Whenever a (released and) unscheduled job $j$ reaches its starting deadline 
(recall that this is $d_j - p_j$), we
activate/schedule an interval $I=RI(j)$ starting at time $d_j - p_j$
of length $e^{\eta+y}$,
where $\eta \in \ZZ$ is smallest such that $p_j \leq e^{\eta+y}$.
The job $j$ is called a {\em flag job}, as in previous works,
and let $F$ denote the set of flag jobs.
All the released and unscheduled jobs that fit
in $I$ are scheduled to start immediately in $I$.
Simultaneous job arrivals and simultaneous starting deadlines are resolved in an arbitrary manner (any deterministic rule will do).
Note that for two flag jobs $j$ and $j'$, we can have $\lambda \left( RI(j) \cap RI(j') \right) > 0 $.

Our algorithm can be implemented to run in $O(n\log n)$ time using two min-heaps over
the released-but-unscheduled jobs. 
The first is keyed by starting deadline
$d_j-p_j$ and determines the next flag event; the second is keyed by processing
time. 
For an arrival at
$t$ one has $d_h - t \ge p_h$ automatically, so $h$ fits into $[a,b)$ iff
$b \ge t + p_h$: maintain the running maximum right endpoint over activated
intervals, $O(1)$ per arrival. At a flag event activating an interval of length
$\ell$ at $t'$, every released unscheduled $h$ has $d_h - t' \ge p_h$, so
``fits'' collapses to $p_h \le \ell$: a priority queue keyed on $p_h$ gives
$O(\log n)$ amortized.

For the analysis, let $C_1, C_2, \ldots$ be the maximal intervals of $\OPT$, an optimum solution. 
Thus the objective of $\OPT$, which we denote by $\opt$, is $\sum_i \lambda(C_i)$.
The algorithm's output costs at most $\lambda \left( \bigcup_{j \in F} RI(j) \right).$
Following \cite{koehler2017busy}, we partition the flag jobs into two groups:  $F_{in}$ is the set  of {\em inside} flag jobs  
$j$ for which there exists $i$ with $RI(j) \subseteq C_i$ and $F_{out}$, the remaining flag jobs. 
It is immediate that 
\begin{equation}
    \label{e_F_in}
    \lambda\left( \bigcup_{j \in F_{in}} RI(j) \right) \leq \sum_i \lambda(C_i) = \opt.
\end{equation}

We further partition $F_{out}$ into sets $F_i$, where $F_i$ is the set of flag jobs
of $F_{out}$ that are scheduled by $\OPT$ in  interval $C_i$.
Consider now the jobs of $F_i$ (for some fixed $i$):  $j_1, j_2, \ldots, j_s$
(here $s = s(i)$), sorted in the order they are scheduled by the algorithm.
We claim that, for all $h = 2, \ldots, s$, we have  $\lambda(RI(j_h)) \geq e \cdot \lambda(RI(j_{h-1}))$
(\cite{koehler2017busy} has a similar situation, with successive such intervals doubling in size).

\begin{figure}[th]
\begin{center}\leavevmode%
\scalebox{1.01}{
  \includegraphics{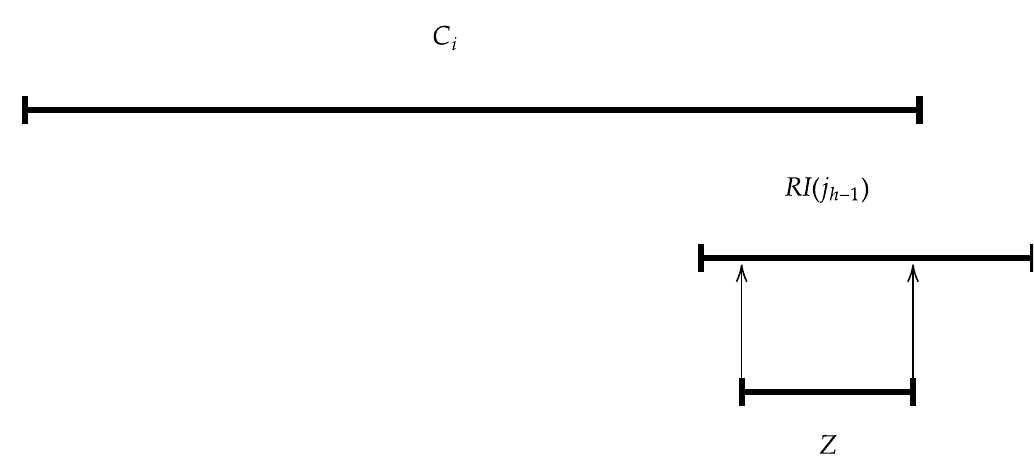}
}
\end{center}
\caption{Here, the interval $RI(j_{h-1})$  with left endpoint $d_{j_{h-1}} - p_{j_{h-1}}$ overlaps but is not included in $C_i$.
If $r_{j_h}$, the release time of flag job $j_h$, is after $d_{j_{h-1}} - p_{j_{h-1}}$,
then $r_{j_h}$ would be the left endpoint of the depicted interval  $Z$, and since job $j_h$ 
can be scheduled in $C_i$,  it can be scheduled in $Z \subseteq RI(j_{h-1})$, contradicting $j_h$ being a flag job.
}
\label{f_double}
\end{figure}

Indeed, note that the right endpoint of $RI(j_{h-1})$ is to the right of $R_i$,
the right endpoint of $C_i$,
since $j_{h-1} \not \in F_{in}$ because it is a flag job from $F_i$
(also, the jobs of $F_i$ can be scheduled inside $C_i$ and thus have a starting deadline at or after the left endpoint of $C_i$).
The left endpoint of $RI(j_{h-1})$ is $d_{j_{h-1}} - p_{j_{h-1}}$.
We cannot have $r_{j_h}  \geq d_{j_{h-1}} - p_{j_{h-1}}$, which we argue as follows:
$r_{j_h} \leq R_i$  (the right endpoint of $C_i$),
since $j_h$ can be scheduled in $C_i$,
so if $d_{j_{h-1}} - p_{j_{h-1}} > R_i$, we are done.
See Figure \ref{f_double} for an illustration of the next argument.
And if $d_{j_{h-1}} - p_{j_{h-1}} \leq R_i$, and $r_{j_h}  \geq d_{j_{h-1}} - p_{j_{h-1}}$, then,
since $j_h$ can be scheduled in $C_i$, we have that $j_h$ can be scheduled in $RI(j_{h-1})$ and 
as a result $j_h$ would not be a flag job.
On the other hand, with $r_{j_h} < d_{j_{h-1}} -  p_{j_{h-1}}$, and using that $j_h$ is scheduled as a flag job
after $j_{h-1}$, the only reason why $j_h$ is not scheduled
in $RI(j_{h-1})$ is that $p_{j_h} > \lambda(RI(j_{h-1}))$. 
Thus $\lambda\left( RI(j_h) \right) > \lambda(RI(j_{h-1}))$,
and by the way these
lengths are assigned by the algorithm (bucketing), 
$\lambda(RI(j_h)) \geq e \cdot \lambda(RI(j_{h-1}))$.

Based on this claim,
\begin{equation}
    \label{e_Ls}
    \sum_{h=1}^s \lambda(RI(j_h)) \leq
    \lambda(RI(j_s)) ( 1 + e^{-1} + e^{-2} + \ldots ) 
    = \frac{1}{1-e^ {-1}} \lambda(RI(j_s)) = \frac{e}{e-1} \lambda(RI(j_s)). 
\end{equation}

The inequality above holds 
for every realization of $s(i)$ and  of $j_{s(i)}$ (the last flag job of $F_i$).
Recall that $y$ is chosen at random by the algorithm uniformly from the interval $[0,1)$.
Let $\hat{\lambda}_i$ (a random variable) be $e^{\eta+y}$ where $\eta \in \ZZ$ is minimum such that $e^{\eta+y} \geq \lambda(C_i)$.
The following computation can be found in \cite{CNS04,SLLA25}.
Let  $x \in [0,1)$  be such that $\lambda(C_i) = e^{k_i+x}$, for integer $k_i$.
Then, if $y \geq x$, we have $\hat{\lambda}_i = e^{k_i+y}$, and if $y < x$, then $\hat{\lambda}_i = e^{k_i+1+y}$.
We have that the expectation of  $\hat{\lambda}_i$ is $\lambda(C_i)$ times
\begin{equation}
\label{e_integral}
    \int_{0}^{x}  e^{y+1-x} dy + \int_{x}^{1} e^{y-x} dy = e^{1-x} (e^x - e^0) + e^{-x}(e^1 - e^x) = e-1.
\end{equation}

Note that $p_{j_s} \leq \lambda(C_i)$ since $\OPT$ schedules $j_s$ in $C_i$ and therefore $\lambda(RI(j_s)) \leq \hat{\lambda}_i$ 
and this is valid for every realization of $s(i)$ and  of $j_{s(i)}$ (the last flag job of $F_i$).
Combining this inequality with Inequality (\ref{e_Ls}) and 
Equality (\ref{e_integral}), we obtain:
\begin{equation*}
    \mathbb{E} \big[ \sum_{h=1}^s \lambda(RI(j_h)) \big] 
    \leq \frac{e}{e-1}  \mathbb{E}[\hat{\lambda}_i] = e \cdot \lambda(C_i).
\end{equation*}
Summing this last inequality over $i$, and recalling Inequality (\ref{e_F_in}), we obtain:
\[
  \mathbb{E}\big[ \lambda\Big(\bigcup_{j\in F} RI(j)\Big) \big]
  \;\le\; \mathbb{E} \big[ \lambda\Big(\bigcup_{j\in F_{in}} RI(j)\Big) \big]
        + \mathbb{E} \big[ \sum_{j\in F_{out}}\lambda(RI(j)) \big]
  \;\le\;\opt+e\cdot\opt = (1+e) \opt.
\]
\end{proof}

The following instance is tight for the auxiliary variant that pays
for every activated  $RI(j)$, for $j$ a flag job, in full.
It does not establish tightness for the standard span implementation,
which pays only for the union of actual job executions.

Precisely, for any $\epsilon > 0$,
we can construct an example where the competitive ratio of the algorithm is at least $1 + e - \eps$, assuming that all the
activated intervals (the intervals $RI(j)$ for flag jobs $j$) 
are kept fully open by the algorithm. See Section \ref{s_concl} for discussion on 
the algorithm that keeps open only the intervals where a job is actually scheduled.

See Figure \ref{f_example} for an illustration.
Let $u$ and $w$ be (large) integers.
We have one interval $C = [0,uw)$ that is going to be the optimum solution.
A job $j$ is called {\em rigid}  if $d_j = r_j + p_j$, and then $I_j = [r_j,d_j)$.
We have $u$ rigid job groups,
each group with  $w-2$ rigid jobs of processing time $1$,
with the release time of the $i^{th}$ job in group $k$ being $(k-1)w + (i-1)$, for $1\le i\le w-2$ and $1\le k\le u$.  
We have one ``flexible" group of  $z :=  \lfloor \ln u \rfloor -1$  jobs, all with release time $0$.
The $h^{th}$ job of the flexible group has deadline $5uw(z +3 - h)$ 
and processing time $uw \cdot e^{1-h}$, for $1 \le h \le z$.
Note that all the flexible jobs have processing time at least $w$ and at most $uw$.
Note that all the jobs fit in the interval $C$, and thus $\opt$  is $uw$.

On the other hand, the randomized online algorithm will schedule all the rigid jobs at a cost
of at least $u(w-2)$ and does not schedule anything
in the nonempty intervals $[(k-1)w + w - 3 + e, kw)$, for $k =  1, 2, \ldots, u$. As such, none of the flexible
jobs, all with processing time at least $w$ and starting deadline bigger than $uw$, will be scheduled in $C$.
Moreover, the flexible jobs will reach their starting deadlines
in increasing order of their processing times, and none will fit in the interval where the previous one is scheduled.
As a result, each flexible job will be scheduled in its own  interval (pairwise disjoint intervals),
and the expected length of this interval is $(e-1)$ times the processing time of the flexible job,
from the computation of Inequality (\ref{e_integral}). With $u$ large enough, Inequality (\ref{e_Ls}) is
within  a multiplicative factor of $(1-\epsilon/8)$ of being tight, and with $w$ large enough, Inequality (\ref{e_F_in}) is also
within a multiplicative factor of $(1-\epsilon/8)$ of being tight. The result is that the expected cost of the algorithm
is at least $(1 + e - \eps) \opt$.

\begin{figure}[th]
\begin{center}\leavevmode%
\scalebox{0.58}{
  \includegraphics{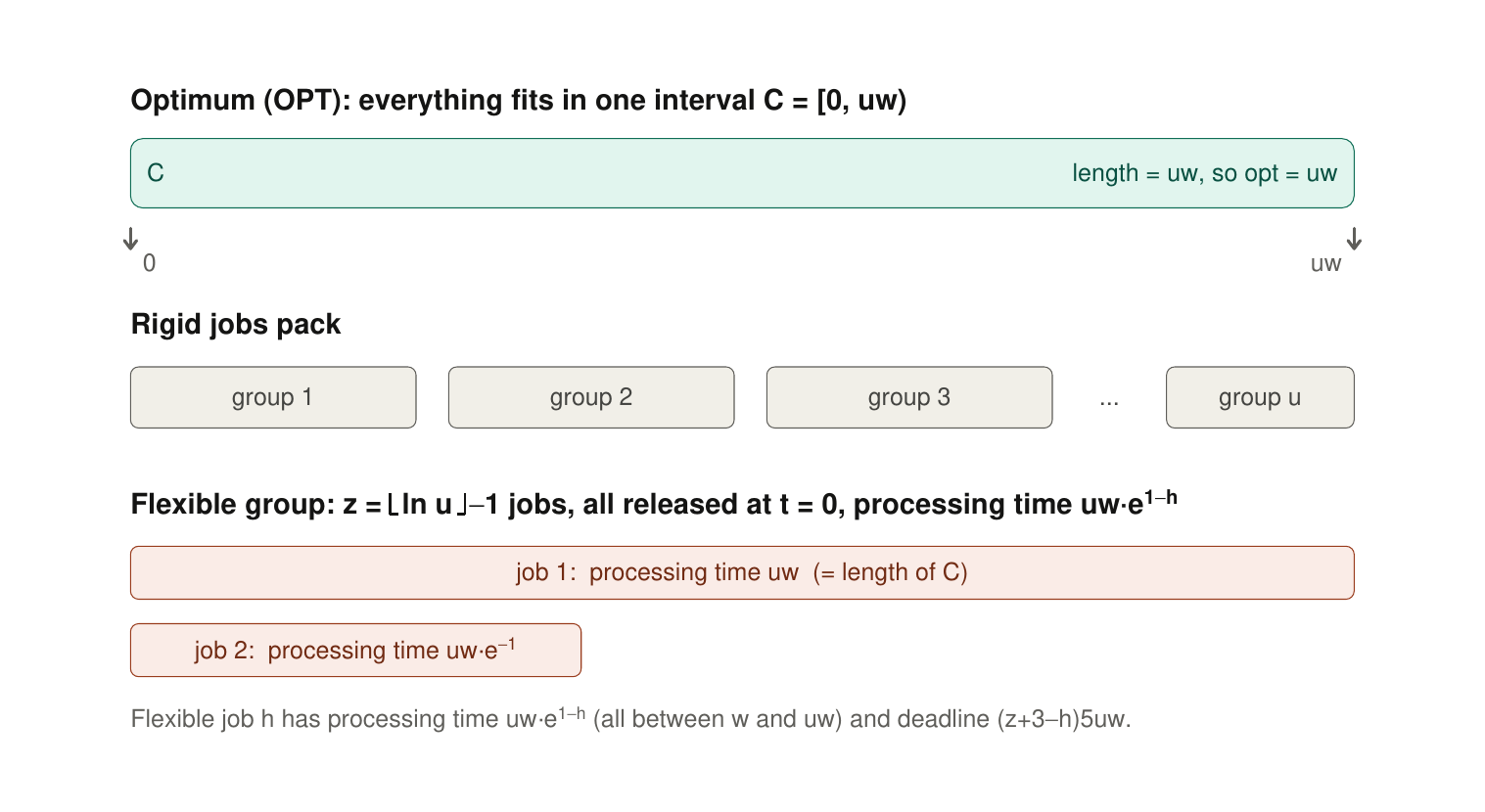}
}
\end{center}
\caption{An illustration of the series of examples showing that the competitive ratio of the auxiliary randomized online algorithm is $1+e$.
The short rigid jobs from the groups will be scheduled very well, but there will be gaps between groups, and because of this
the algorithm will  schedule the long flexible jobs  to start each at their  starting deadlines 
(not drawn, but staggered so that the shorter flexible jobs have smaller deadlines). 
}
\label{f_example}
\end{figure}

\section{Algorithms for Capacitated Busy Time with  Lookahead}
\label{s_cap}

Recall that $g$ is the capacity of a machine, the number of jobs it can run simultaneously. We use $\opt_c$ for the optimum of the capacitated problem,
and use later that $\opt_c \geq \opt$ 
(so $\opt$ still is the minimum span of the input).
First we mention the ``volume"\footnote{Also called ``mass" in the literature.}
lower bound:
\begin{equation}
    \label{e_v}
    \opt_c \geq \frac1g v(J),
\end{equation}
where for a set of jobs $A$, we define $v(A) := \sum_{j \in A} p_j$,
and $J$ is the input set of jobs.

If $g \leq 2$, one gets a $2$-competitive algorithm by simply scheduling every job on its own machine.
From now on in this section, we assume $g \geq 3$.
We present two online algorithms (each having two variants, deterministic or randomized). 

\cite{koehler2017busy} succinctly describe a $12$-competitive online algorithm for Capacitated Busy Time scheduling with a $2 p_{max}$-lookahead.
Next we use a variant of their algorithm,
with three straightforward modifications. First, we use our randomized
algorithm from Theorem \ref{thm: randomized}.
Second, we use a  slightly different way of partitioning jobs into  two 
(even and odd) sets, as we only allow half the lookahead.
Third, an easy generalization gives:

\begin{theorem}
    \label{t_l}
    Let $l$ be a positive integer. With a lookahead of $\px/l$ and against an oblivious adversary,
    there is a randomized online algorithm for Capacitated Busy Time scheduling with
    competitive ratio of at most $2 + (l+1)(1+e)$. 
\end{theorem}

\begin{proof}
We discuss the case $l=1$, with the generalization to larger $l$ straightforward
as discussed in the last paragraph of the proof.

Recall the set $F$ of flag jobs of the algorithm from Theorem \ref{thm: randomized},
and that $RI(j)$ is the interval associated with flag job $j$.
Let $f_t$ be the last of the flag jobs scheduled to start before time $t$, if any.
Note that the right endpoint of $RI(f_t)$ exceeds the right endpoint of $RI(j)$ 
for any other flag job $j$ scheduled to start before time $t$, 
since if not, job $f_t$ would be scheduled in $RI(j)$.
Precisely, at time $k p_{max}$,
we start with an interval $I^k$, possibly empty, where  
$I^k = [k p_{max}, + \infty) \cap RI(f_{k p_{max}})$ if $f_{k \px}$ exists, 
and $I^k = \emptyset$ otherwise.
We can see all the jobs with release time before $(k+1) p_{max}$.
We can determine the subset $A_k$ of these jobs that will
be started during the interval $[k p_{max}, (k+1) p_{max})$ by the algorithm from Theorem \ref{thm: randomized}
(in other words, the span algorithm is simulated one block ahead).
These are either flag jobs, or jobs that are to be scheduled in the intervals determined by the flag jobs.
The next paragraph explains why $A_k$ can indeed be computed at time $k p_{max}$. 

Note that even if some of the intervals $RI(j)$ for $j$ in $F$ have non-empty pairwise intersection, a non-flag job $j'$ fits
fully in one such $RI(j)$: precisely we
assign $j'$ to the earliest-created flag interval $RI(j)$ satisfying
$\lambda(RI(j)\cap W_{j'})\ge p_{j'}$
(This tie-breaking convention is a minor modification that does not affect the analysis of Theorem \ref{thm: randomized}).
Let $U_k$ be the set of jobs that have not yet been scheduled by  the
algorithm from Theorem \ref{thm: randomized} at time $k \px$,
and that are to be released before $(k+1) \px$. 
Due to lookahead, $U_k$ is known to this algorithm at time $k \px$.
At time  $k p_{max}$, we have enough information to see
which jobs  of $U_k$ fit inside $I^k$, which jobs  of $U_k$ become flag jobs, and which jobs of $U_k$ fit inside $RI(j)$ 
for some $j$ that becomes a flag job before $(k+1) \px$.
All these jobs together become the set $A_k$.

Also note that  we can compute $I^{k+1}$ as well:
the latest flag job  $j \in \bigcup_{h \leq k} A_h$
has an interval whose overlap with the next block is at least as large as that of every earlier flag interval
(because flag job $j_1$ cannot fit in $RI(j_2)$ if $j_2$ is a flag job scheduled before $j_1$).

The execution interval of each of the jobs in $A_k$ is fixed by the algorithm from Theorem \ref{thm: randomized}
and scheduled {\em on new machines} by the $3$-approximation
offline  Algorithm 1 (and Theorem 5) of \cite{CKM17}.
If $I_j$ is fixed for $j \in A$, we define  $\Span(A) = \bigcup_{j \in A} I_j$.
The busy time of the output of this algorithm is at most
$\lambda(\Span(A_k)) + (2/g) v(A_k)$,
from the first two sentences of the proof of their Theorem 5. 
Note that the sets $A_k$ are disjoint for different values of $k$.
The overall expected busy time is then at most:
\begin{equation}
    \label{e_even_odd}
    \mathbb{E} \big[ \sum_k \left( \lambda\left( \Span(A_k) \right) + (2/g) v(A_k) \right) \big] =
    (2/g) v(J) +  \sum_{k \mbox{ even}} \mathbb{E} \big[ \lambda \left( \Span(A_k) \right) \big]  
    +  \sum_{k \mbox{ odd}} \mathbb{E} \big[ \lambda \left( \Span(A_k) \right) \big].
\end{equation}
Note that $\Span(A_k)$ is disjoint from $\Span(A_{k+2})$, and therefore $\sum_{k \mbox{ even}} \lambda \left( \Span(A_k) \right)$
is at most the output of the algorithm from Theorem \ref{thm: randomized}.
The same holds for $\sum_{k \mbox{ odd}} \lambda \left( \Span(A_k) \right)$.
As $\opt_c$ is at least the optimum span of $J$, using Theorem \ref{thm: randomized},  and inequalities (\ref{e_v}) and (\ref{e_even_odd}),
we obtain that the cost of the output of this (randomized online with lookahead) algorithm is, in expectation, no more than $(4+2e)\opt_c$.

With a lookahead of $\px/l$, for integer $l > 1$, 
one can use the same approach, considering intervals of the form $[k \px/l, (k+1) \px/l)$.
In the analysis above, we have that $\Span(A_k)$ is disjoint from
$\Span(A_{k+l+1})$, and
instead of even and odd, we have $l+1$   groups with 
a randomized competitive ratio of at most $2 + (l+1)(1+e)$.
\end{proof}

If, above,  we use the variant with the deterministic online algorithm of
\cite{koehler2017busy} instead of our randomized algorithm, we
get a competitive ratio of at most $2 + 5(l+1)$. For $l=1$,
this is the $12$-competitive algorithm from Section 2.3 of \cite{koehler2017busy}
except we are careful to only allow a lookahead of $p_{max}$.

From now on the lookahead is $p_{max}$. A deterministic
$9$-competitive algorithm was reported by
\cite{albers2025onlinebusytimescheduling} (v2),
but their proof fails. Appendix  \ref{a_albers}
discusses their algorithm and gives a counterexample.
We could not fix their proof, and instead use 
a variant of their algorithm carefully combined with
ideas from \cite{koehler2017busy}, \cite{CKM17},
and looking at the problem from a different angle.
With a better analysis, we obtain  below
a deterministic competitive ratio of at most $8$.

\subsection{$3$-Tracks and the inclusion property with $p_{max}$-lookahead}
\label{ss_3tracks}

The algorithm $\OGT$ from this subsection replaces
Algorithm 4 of \cite{albers2025onlinebusytimescheduling} (v2).
In this subsection, the jobs are rigid,
as a result of simulating online algorithms for Span Minimization as explained later.
Precisely, $r_j$ and $d_j$ in this
subsection refer to endpoints of these rigidized execution intervals,
rather than for the original flexible window.
We use our own notation, which we developed in the proof of Theorem \ref{t_l}.
We say that a rigid job is {\em active} at time $t$ if $t \in [r_j,d_j)$.
With $J$ being rigid, $I_j$ is fixed from the input, and for $A \subseteq J$,
we define $\Span(A) := \bigcup_{j \in A} I_j$.

A {\em track}, following \cite{CKM17}, is a set of pairwise w-disjoint jobs.
A {\em $3$-Track} consists of three tracks.
Algorithm $\OGT$ produces a list of $3$-Tracks
$T_1, T_2, \ldots$, where the order matters.

Before running Algorithm $\OGT$,
initialize $\tau=0$, all $3$-Tracks and tracks as empty, and then
Algorithm $\OGT$ is invoked for $k=0,1,2,\ldots$.

A {\em chain} is a set of jobs such that at any moment in time, no more than two are active.
From a set of jobs $A$, the procedure $LongChain(A)$ that we describe below finds a chain $P$ such that $\Span(P) =  \Span(A)$.
Add to $P$ the job $j$ with the earliest release time.  Set $j$ as the ``last" job in the chain.
If there exists a job $j'$ that intersects the last job in the chain and $j'$ has deadline after the last job of the chain, find among these jobs the one
with the largest deadline,
add it to the chain, and make it the last job of the chain.
If no such $j'$ exists, find $\hj$, the job with 
the earliest release time among the jobs
with release time at or after the deadline of the last job of the chain,
breaking ties to favor the job with the latest deadline.
Add $\hj$ to the chain and make it the last job of the chain. If $\hj$ does not exist, stop.

It is not hard to verify that $LongChain(A)$ has the desired properties.
It is also part of the proof of Theorem 5 of \cite{CKM17}.
Starting with $A \leftarrow J$ and then
repeatedly using $ P \leftarrow LongChain(A)$ and 
updating $A \leftarrow A \setminus P$ would give us $2$-Tracks that satisfy the inclusion property mentioned in Appendix \ref{a_albers}.
However, this would not be an online algorithm with limited lookahead.

Instead we proceed with blocks of size $p_{max}$ as in Section \ref{s_cap},
and at time $k\px$ we must schedule the set $A_k$.
We start the block at time $k p_{max}$  with a list of $3$-Tracks  
$T_1, T_2, \ldots, T_{\tau}$.

As a $3$-Track is being produced, it is naturally split into tracks:
sort the jobs by starting time, and just don't overlap existing jobs
on track. Already placed jobs are not moved.
Pick a fixed deterministic rule if there are ties in placing this job.
For $3$-Track $T_i$, we call $T_i(1)$, $T_i(2)$, and $T_i(3)$ these three tracks.
Let $P^k_i(1)$,  $P^k_i(2)$,  and $P^k_i(3)$ be the set of jobs 
added to  $T_i(1)$, $T_i(2)$, and $T_i(3)$ when processing the block
$[k \px, (k+1) \px)$.
To  be clear, the sets  $P^k_i(1)$,  $P^k_i(2)$,  and $P^k_i(3)$ 
are {\em after} 
the {\bf for} loop removes jobs to eliminate the fourth job active at some time $t$.
Please refer to the pseudocode of $\OGT$ in Algorithm \ref{alg:v3}.

\begin{algorithm}[hbt!]
\caption{\alg{OnlineGreedyTrackingV2}}\label{alg:v3}
$A \leftarrow A_k$ \\
  \For{ $ i \leftarrow 1$ to $\tau$ }{
        \If{$A = \emptyset$ }{RETURN}
        $P \leftarrow LongChain(A)$ \\
        \For{$t'$ starting with $k p_{max}$, increasing, and checking every job interval endpoint}{
            \If{there are $4$ jobs of $T_i \bigcup P$ scheduled at time $t'$}{
            remove from $P$ the job with the earliest deadline from the jobs of $P$ active at time $t'$
            }
        }
        $T_i \leftarrow T_i \bigcup P$ \\
        $A \leftarrow A \setminus P$ \\
        Put $P^k_i(1)$ on track $3i-2$ \\
        Put $P^k_i(2)$ on track $3i-1$ \\
        Put $P^k_i(3)$ on track $3i$ \\
  }
  $i \leftarrow \tau$ \\
  \While{$A \neq \emptyset$}{
    $P \leftarrow LongChain(A)$ \\
    $i \leftarrow i+1$ \\
    Create a new $3$-Track $T_i$  with the jobs of $P$ \\
    $A \leftarrow A \setminus P$ \\
        Put $P^k_i(1)$ on track $3i-2$ \\
        Put $P^k_i(2)$ on track $3i-1$ \\
        Put $P^k_i(3)$ on track $3i$ \\
    }
    $\tau \leftarrow i$
\end{algorithm}

The algorithm maintains the following four invariants:
\begin{enumerate}
\item Each $T_i$ is indeed a $3$-Track, meaning that its set of jobs can be 
partitioned into three tracks.
\item 
$\Span(T_i) \subseteq \Span(T_{i-1})$ for $i>1$. 
    \item 
All the jobs with starting time before time $k p_{max}$ belong to exactly
one of these $3$-Tracks. 
\item
Each $T_i$ has at most {\bf two} jobs active at time $k p_{max}$.
\end{enumerate}
These four invariants will also hold at time $(k+1) p_{max}$,
as it is easy to verify, based on the remainder of this paragraph.
The crux is that when a job $x$ is removed
from $P$ because there are four jobs of $T_i \cup P$ active at time $t$,
$\Span(T_i \cup P)$ does not change. Indeed, at time $t$ we had two jobs
from $T_i$, and besides $x$ another job $y$ from $P$, where the
deadline of $y$ exceeds the deadline of $x$. The part of $x$ at or before $t$ 
is included in $\Span(T_i)$ and the part of $x$ 
after $t$ is included in $\Span(\{y\})$.

Note that $\Span(T_i) \subseteq \Span(T_{i-1})$, the inclusion property
that Algorithm 4 of \cite{albers2025onlinebusytimescheduling}  (v2) fails to maintain.

Let $B_1, B_2, \ldots$ be the bundles produced by the algorithm,
sorted by the $3$-Tracks they contain.
Bundle $h$ consists of the tracks $(h-1)g + 1, (h-1)g + 2, \ldots, hg$,
and these are the tracks assigned to machine $h$.
Machine $h$ is busy for the span of  the jobs assigned to it.

The tracks of Bundle $i$ come from $3$-Tracks $T_{u}$ for some range of $u$ that
thankfully we do not need to compute explicitly 
(though surely it can be done; the range formula
depends on $g \pmod{3}$ and $i \pmod{3}$). 
Let $u(i)$ be the first among these $3$-Tracks.
Let $D(i,q)$ be the number of tracks in bundle $i$ that come from $3$-Track $q$.
Then $D(i,q) \in \{0,1,2,3\}$ and $u(i) = \min\{h : D(i,h) > 0\}$.

From the $3$-Track inclusion property, for $q \leq u(i)$ we have
\begin{equation*}
    \lambda ( \Span (B_i) ) \leq \lambda ( \Span (T_{u(i)}) ) 
    \leq \lambda ( \Span (T_q) ) \leq v(T_q)
\end{equation*}
Note that any track in $B_{i-1}$ comes from a $3$-Track $q$ with $q \leq u(i)$.
Take such a $q$. We have
\begin{equation*}
    D(i-1,q) \lambda ( \Span (B_i) ) \leq D(i-1,q) v(T_q)
\end{equation*}

Every bundle other than the last one is full,
and therefore bundle $(i-1)$ is full whenever bundle $i$ is nonempty.
Using this, and
summing the above inequalities for all $q$ with $D(i-1,q) > 0$ we obtain:
\begin{equation*}
    g \lambda ( \Span (B_i) ) \leq \sum_q D(i-1,q)  v(T_q)
\end{equation*}
where we can take the last summation over all $q$. Then:
\begin{equation*}
    \lambda (\Span (B_i) ) \leq \frac1g  \sum_q D(i-1,q)  v(T_q)
\end{equation*}
Summing up over $i>1$ we obtain:
\begin{align}
    \sum_{i > 1} \lambda (\Span (B_i) ) & \leq &
    \frac1g \sum_{i>1} \sum_q D(i-1,q)  v(T_q)  \\
    & \leq &\frac{1}{g} \sum_q v(T_q) \sum_{i>1}  D(i-1,q)  \\
    & \leq &  \frac{1}{g} \sum_q 3 v(T_q)  \leq 3 \frac{1}{g} v(J) \leq 3 \opt_c,
    \label{e_3opt}
\end{align}
with the last inequality from Lower Bound (\ref{e_v}).
The analysis of the deterministic algorithm proceeds as follows:
Simulate the
deterministic online algorithm for Span Minimization of \cite{koehler2017busy}
to obtain the block sets $A_k$ (this simulation
can be done just like the corresponding simulation of our randomized algorithm). 
Since this deterministic algorithm has competitive ratio of $5$,  we obtain:
\begin{equation*}
    \lambda(\Span(B_1)) \leq 5 \opt.
\end{equation*}
Combine this with Inequality (\ref{e_3opt}) and obtain an $8$-competitive algorithm.
For a randomized algorithm, combine Inequality (\ref{e_3opt}) 
with the upper bound of $(1+e) \opt$ on
the span of $B_1$ from Theorem \ref{thm: randomized} to obtain a $(4+e)$-competitive 
online algorithm. 

When the input consists of rigid jobs, $\lambda(\Span(B_1)) \leq \opt$.
Theorem \ref{t_cap} follows.

\section{Lower Bounds on the Competitive Ratio - General Case} \label{sec: adversary-4}

In this section we prove Theorems \ref{thm: lb_nested_deterministic} and \ref{thm: lb_nested_randomized}.
We use results and techniques from \cite{SLLA25}, with non-trivial adaptations for our problem.
First, we follow \cite{SLLA25} by defining an online \emph{Button  problem}.

In this problem, we are given an ordered list $b$ of $m$ buttons where each button $j$ is associated with a price $b_j \in \QQ_{\geq 1}$.
The prices are nondecreasing: $1 = b_1 \leq b_2 \leq b_3 \leq \cdots \leq b_m$.  Some buttons are designated as target buttons, which form a
suffix of the button list, i.e., there exists some $j^* \leq m$ such that buttons $j^*$ to $m$ are all targets, and
none of the  buttons before $j^*$ are targets. The list $b$ is known from the start.
However, we do not know in advance the first target button
(i.e., $j^*$). We can learn whether a button $j$ is a target or not only by pressing the button at the price
of $b_j$. The online algorithm repeatedly presses buttons until it presses any target button, and the natural
objective is to minimize the total price.

Restating the definition from the introduction,
we say that a deterministic algorithm $\cA$ is $\gamma$-competitive if for all lists $b$ and all
$j^* \in \{1, 2, \ldots, m\}$, we have $\cost(\cA,b,j^*)  \leq \gamma \cdot b_{j^*}$,
where $\cost(\cA,b,j^*)$ is the total price that the algorithm $\cA$ incurs on input $b$
until it presses a target button when the first target button is $j^*$.
Note that algorithm $\cA$ can be identified with a sequence  $i_1(\cA,b), i_2(\cA,b), \ldots, i_{k(\cA,b)}(\cA,b)$
which are the buttons pressed by the algorithm  when the target is $m$,
where we assume without loss of generality that $i_h(\cA,b) < i_{h+1}(\cA,b)$
and $i_{k(\cA,b)}(\cA,b) =  m$.
The algorithm, of course, stops after hitting the first target, i.e., the first $w = w(\cA,b,j^*)$
for which $i_w(\cA,b) \geq j^*$, and its total price is 
$\cost(\cA,b,j^*) = \sum_{h=1}^w b_{i_h(\cA,b)}$.

Combining Lemma 5.9 and Theorem 5.11 of \cite{SLLA25}   into a single theorem, we have:
\begin{theorem}\label{thm: lb_button_4} [from Section 5.3 of \cite{SLLA25}]
    For any $\epsilon > 0$ there is an $m = m(\epsilon)$ and a list $b$ of button prices $ 1 = b_1 < b_2 < \cdots < b_m$ such that
    for any deterministic algorithm $\cA$ there exists a $j^* = j^*(\epsilon,\cA)$ such that
    $\cost(\cA,b,j^*) > (4 - \epsilon) b_{j^*}$. Rephrased:
    For any $\epsilon > 0$ there is an $m = m(\epsilon)$ and a list $b$ of button prices $ 1 = b_1 < b_2 < \cdots < b_m$ such that
    for any increasing sequence $S$ of button indices $i_1, i_2, \ldots, i_k = m$,
    there exists a  $j^* = j^*(\epsilon,S)$ such that
    $ \sum_{h=1}^w b_{i_h} > (4 - \epsilon) b_{j^*}$, where
    $w = \min \{u : {i_u} \geq j^*\}$.
\end{theorem}

This theorem implies that no deterministic algorithm for the Button problem
has a competitive ratio smaller than 4.
The statement of the theorem is stronger than this, however, as the list $b$ does not depend on the algorithm $\cA$; 
we will use this fact, 
namely, that for all algorithms, we have  the same $m$ and $b_m$.
As an aside, the list $b$ is given by $b_i = 1 + (i-1) \cdot \delta$, for a small
$\delta$ that depends on $\epsilon$.
See \cite{SLLA25} for details. Earlier work by \cite{zhang2011ski} proves a very
similar statement for Multi-Option Ski Rental,  and they explicitly notice that the same list of options works
against all deterministic algorithms. We prefer to use \cite{SLLA25}  instead of \cite{zhang2011ski} for simplicity, and
for its discussion on randomized algorithms.

A $4$-competitive deterministic algorithm for the Button  problem follows from the works of \cite{zhang2011ski,SLLA25},
and also it is straightforward to
verify that  the following algorithm achieves this ratio:
in iteration $i$, press the  button 
$j:= \arg\max\{h : b_h\le 2^i\,\}$.

Let $f$ denote the lookahead function.

\begin{proof}[\emph{of Theorem \ref{thm: lb_nested_deterministic}}]
   Let $\epsilon > 0$.  We pick a large integer $N = N(\eps)$ whose exact value will be fixed later.
    Let $m = m(\eps/4)$, and let list $b =b (\eps/4)$ be from Theorem \ref{thm: lb_button_4}, creating
    an instance of the Button problem.

    Let $\cA$ be an arbitrary deterministic online algorithm for Span Minimization.
The adversary will make $N$ {\em mega-releases}, each consisting of $m$ jobs arriving at the same time and
with different processing times and deadlines, with fixed gaps between the deadlines of jobs from the same mega-release. 
We will give the jobs' processing times and deadlines soon.  
Let $\beta = (2+m) ( 2+f(b_m))$.
We construct disjoint intervals, which we call {\em phases}, that are helpful in the proof. 
Phase 0, defined to shorten definitions,  is the interval $[0, \beta b_m   (2+m)^{N-1})$. 
Phase $q$ will be an interval $[L_q,R_q)$.

See Figure \ref{fig: lb4} for an illustration.
For $q =  1, 2, \ldots, N$, the adversary executes iteration $q$ of the following procedure.
At $R_{q-1}$ (exactly the end of phase $q-1$),
the adversary initiates mega-release $q$.
We also have $L_q = R_{q-1}$.
For $i \in \{1,\ldots,m\}$, the $i^{th}$ job of mega-release $q$ has deadline
$$ \beta b_m \left( (i + 1) (2 + m)^{N-q} + \sum_{h = 1}^{q-1}  (1 + j_h^*) (2+m)^{N-h} \right), $$
so the gap between the deadlines of the $i^{th}$ and $(i+1)^{th}$ jobs in mega-release $q$ is 
$ \beta b_m (2 + m)^{N-q} $.
For $q > 1$,
this is also the gap between the deadline of the last job of mega-release $q$ and the 
$(1 + j_{q-1}^*)^{th}$ job of mega-release $q-1$ (if $j_{q-1}^* < m$),
and the gap between the 
$(j_{q-1}^*)^{th}$ job of mega-release $q-1$ and the start of phase $q$.
Also,  for $i \in \{1,\ldots,m\}$, the $i^{th}$ job of mega-release $q$ has processing time $b_i$ 
(the same  $b_i$ for all mega-releases, the same list of buttons).
In iteration $q$, the adversary also chooses $j_q^*$ as explained below.
 Phase $q$ consists of the time interval from the arrival of mega-release $q$  to
 $R_q = \beta b_m  \left( (2+m)^{N-q-1} +  \sum_{h = 1}^{q}  (1 + j_h^*) (2+m)^{N-h} \right) $.

We say that  $\cA$  {\em books} interval $I$ for job $j$ if $I$ is the
final maximal interval
from the schedule of $\cA$ where job $j$ is completely executed by the algorithm.
We use this definition since one may think that $\cA$ can ``extend" intervals by adding to their end.
We say that a button $i$ corresponds to an interval $I$ with $\lambda(I) \geq 1$ if $i$ is largest with $b_i \leq \lambda(I)$.
All the jobs released by the adversary  will have processing time at least $1$, and  therefore
$\cA$ never books an interval of length smaller than $1$
($\cA$ could extend an existing scheduled interval by an amount smaller than $1$).
Note that if the algorithm $\cA$ starts executing a job $j$, and then it aborts and restarts $j$
(which $\cA$ is allowed to do in this section), then $\cA$ books $j$ in the interval where $j$ is completed.

In iteration $q$, the adversary
constructs from $\cA$ a deterministic algorithm $\cA_q$ for the Button problem with input list $b$.
Even if $\cA$ scheduled in phases $0, 1, \ldots, q-1$
all jobs of mega-releases $1, 2, \ldots, q-1$, it still needs to schedule the jobs of mega-release $q$.
Specifically, $\cA_q$ is constructed based on what $\cA$ does in the time 
from the arrival of mega-release $q$ until the moment it
finishes booking all jobs of mega-release $q$ (assuming no other jobs are released), as follows:
\begin{itemize}
    \item Let $I$ be the interval where $\cA$ books the first
job of mega-release $q$, and let $I' = I \cap [L_q,\infty)$.
Note that $\lambda(I') \geq 1$.
$\cA_q$ starts by pressing the button that corresponds to $I'$.
    \item Thereafter, whenever algorithm $\cA$ books a new interval $I$,
    and the button index $i$ corresponding to $I$
    is larger than all the indices previously used, then
     $\cA_q$ presses button $i$.
\end{itemize}

Simulating $A$ on the fixed set of jobs of mega-release $q$,
assuming no further releases, determines an
increasing sequence of button indices.
Any such sequence defines a deterministic Button algorithm, known to the adversary.

At the deadline of job $i$ of mega-release $q$, assuming mega-release $q+1$ did not happen yet, $\cA$ cannot see 
any jobs other than those already released, since the gap between the current time
and the next possible release is at least $\beta b_m > f(b_m) = f(\px)$.
Therefore lookahead does not help $\cA$ in any way when booking the jobs of mega-release $q$.

As a result, $\cA_q$ is a valid online algorithm for the Button problem.
$\cA_q$ does not need to be the same algorithm as $\cA_{q'}$, as $\cA$ can treat mega-releases differently.
Nevertheless,
$\cA_q$ is a deterministic algorithm for the same Button instance, and the adversary knows  how $\cA_q$ operates
after making mega-release $q$.
From Theorem \ref{thm: lb_button_4}
there exists a $j^*_q$ for which $\cost(\cA_q,b,j^*_q) > (4-\eps/4) b_{j^*_q}$.
This finishes the description of the adversary's choices in iteration $q$.

Note that no job in mega-release $q$
can be scheduled in any of phases $0, 1, \ldots, q-1$.
Also, for $0 < i \leq q-1$, we have 
$   \sum_{h = 1}^{q-i}  (1 + j_h^*) (2+m)^{N-h}  < 
  (2 + m)^{N-q} +   \sum_{h = 1}^{q-1}  (1 + j_h^*) (2+m)^{N-h} $ and 
$  (m+1) (2 + m)^{N-q} +   \sum_{h = 1}^{q-1}  (1 + j_h^*) (2+m)^{N-h}  < ( 1 + 1 + j_{q-i}^*) (2+m)^{N-(q-i)} +
 \sum_{h = 1}^{q-i-1}  (1 + j_h^*) (2+m)^{N-h}$. 
 Therefore
every job in mega-release $q$ is either nested in or w-disjoint from every job of mega-release $q-i$.

\begin{figure}[th]
\begin{center}\leavevmode%
\scalebox{0.88}{
  \includegraphics{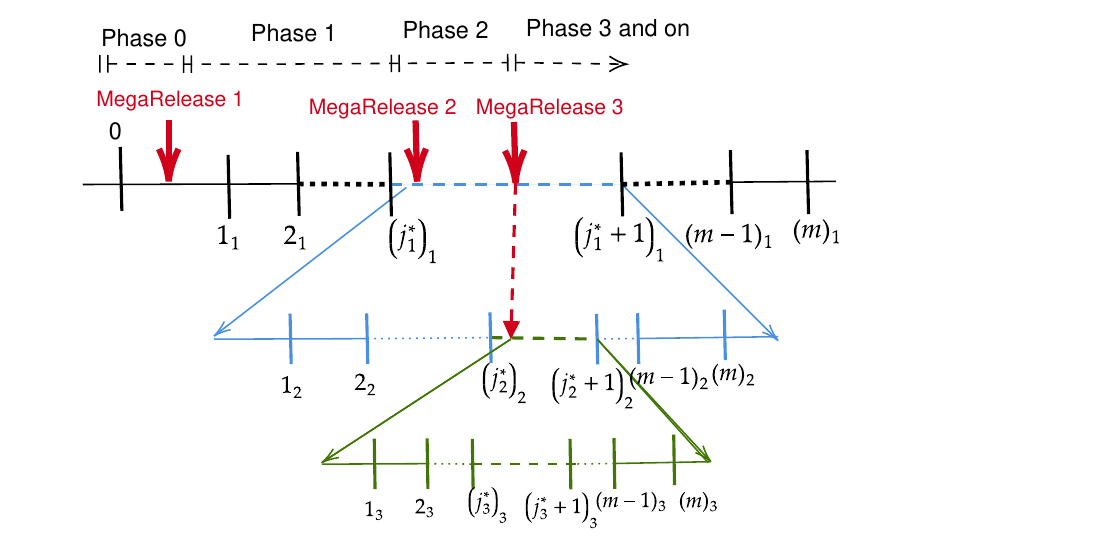}
}
\end{center}
\caption{Not drawn to scale. An illustration of the construction.  We use $i_q$ to depict the deadline of the $i^{th}$ job of mega-release $q$.
All the jobs of mega-release 2 are nested between
the deadline of the $({j^*_1})^{\text{th}}$ job of mega-release 1 and the 
deadline of  job $1+{j^*_1}$ from mega-release 1,
except if $j^*_1 = m$, when  all jobs of mega-release 2 are released
after the deadlines of all jobs of  mega-release 1. The same
holds for the jobs of mega-release 3, nested
between two consecutive deadlines of the jobs of mega-release 2 (with the same exception).
The processing times of the jobs are not depicted (the processing time of the $i^{th}$ job of mega-release $q$ is $b_i$ for all $q$), 
and the gap between the deadlines is sufficiently large for
any job to be fully processed between any two distinct deadlines.
}
\label{fig: lb4}
\end{figure}

We achieved that  the $N+1$ phases occupy disjoint intervals.
We also obtained deterministic algorithms $\cA_1, \cA_2, \ldots, \cA_N$ for the Button problem with input $b$,
and integers $j^*_1, j^*_2, \ldots, j^*_N$ such that, for $q = 1, 2, \ldots, N$, we have:
\begin{equation}
    \label{e_bj}
 \cost(\cA_q,b,j^*_q) > (4-\eps/4) \cdot b_{j^*_q}.
\end{equation}

We claim that
the cost of $\cA_q$ on Button instance  $b$ with $j^*_q$ as the first target
is at most the cost of $\cA$ in phase $q$. Indeed, $\cA$ must book the first $j^*_q$ jobs of mega-release $q$ 
and thus needs an interval $\tI_q$ of length at least $b_{j^*_q}$ that starts no later than the starting deadline of the $(j^*_q)^{th}$ job
of mega-release $q$. Each of the buttons $\cA_q$ 
presses when the target is button $j^*_q$,
except possibly the last one, corresponds to a distinct interval  
that $\cA$ uses in phase $q$, 
and the cost of the button is at most the length of the corresponding interval.
The last button corresponds to an interval we call $\tI'_{q}$.
If $\tI'_{q}$ is fully contained in phase $q$, we are done.
If $\tI'_{q}$ is so long that it enters phase $q+1$, then the portion of $\tI'_{q}$ inside phase $q$ has length at least $b_m$
(due to the gap between the deadline of the 
$(j^*_q)^{th}$ job of mega-release $q$ and the right endpoint of phase $q$), and indeed $\tI'_{q}$ corresponds to button $m$.
In all cases, $\cA_q$ does press a button numbered at least $j^*_q$.

Using Inequality (\ref{e_bj})
we obtain that the cost of $\cA$ in phase $q$ exceeds $(4 - \eps/4) \cdot b_{j^*_q}$.
By summing up over the disjoint phases, we obtain:
\begin{equation}\label{eq_sigma}
   \cost(\cA)  > (4 - \eps/4) \cdot \sum_{q=1}^N  b_{j^*_q}.
\end{equation}
Now we construct another (offline) solution $S$ to this  instance of Span Minimization:
for all $q = 1, 2, \ldots, N$,
use one interval of length $b_{j^*_q}$ at the beginning of phase $q$, and after the end of phase $N$, use
one interval of length $b_{m}$. Note that
the interval of length $b_{j_q^*}$ at the beginning of phase $q$
        schedules the first $j_q^*$ jobs of mega-release $q$.
        Also note that  every remaining job's window contains the common final interval of length $b_m$.
Therefore $S$ indeed schedules all the jobs of this instance.

Now we pick $N \geq \lceil \frac{8}{\eps} b_m \rceil$.
Since $b_{j^*_i} \geq b_1 = 1$, we get $N \leq \sum_{i=1}^N b_{j^*_i}$.
Using this and algebraic manipulation, we obtain that the cost of $S$
is at most $\sum_{i=1}^N b_{j^*_i} ( 1 + \eps/8)$.
Using Inequality (\ref{eq_sigma}) and algebraic manipulation, we obtain
\[ \cost(\cA) > (4 - \eps) \cost(S),\]
  completing the proof of Theorem \ref{thm: lb_nested_deterministic}.
\end{proof}

While the previous theorem is not new (see \cite{LT24}) and the proof above may be a complete version of
what the version v1 of \cite{albers2025onlinebusytimescheduling} tries to do,
its method allows us to quickly obtain  the proof of Theorem \ref{thm: lb_nested_randomized} (a new result) below.

Restating the definition from the introduction,
we say that a randomized algorithm $\cA$ for the Button problem is 
$\gamma$-competitive if for all lists $b$ and all
$j^* \in \{1, \ldots, m\}$, 
$\mathbb{E}[\cost(\cA,b,j^*)]  \leq \gamma \cdot  b_{j^*}$,
where $\cost(\cA,b,j^*)$ is the total price that the algorithm $\cA$ incurs on input $b$
until it presses a target button, when the first target button is $j^*$.
Note that a randomized algorithm $\cA$ can be identified with a distribution over
deterministic algorithms, and thus a distribution over sequences defined earlier for
deterministic algorithms.

Recall that $e$ is Euler's number. From the proof of Theorem 5.2 of \cite{SLLA25}, we have:
\begin{theorem}\label{thm: lb_button_e} [from Section 5.2 of \cite{SLLA25}]
    For any $\epsilon > 0$ there is an $m = m(\epsilon)$ and a list $b$ of button prices
    $ 1 = b_1 < b_2 < \cdots < b_m$ such that
    for any randomized algorithm $\cA$ there exists a $j^* = j^*(\epsilon,\cA)$ such that
    $\mathbb{E}[\cost(\cA,b,j^*) ]> (e - \epsilon) \cdot b_{j^*}$.
    Rephrased:
    For any $\epsilon > 0$ there is an $m = m(\epsilon)$ and a list $b$ of button prices $ 1 = b_1 < b_2 < \cdots < b_m$ such that
    for any distribution $\Pi$ over increasing sequences $S$ of button indices
    $i^S_1, i^S_2, \ldots, i^S_{k(S)} = m$,
    there exists a  $j^* = j^*(\epsilon,\Pi)$ such that
    $ \mathbb{E}_{S \sim \Pi} \big[  \sum_{h=1}^{w(S)} b_{i^S_h} \big] > (e - \epsilon) b_{j^*}$, where
    $w(S) =  \min \{u : {i^S_u} \geq j^*\}$.
\end{theorem}
This theorem implies that no randomized algorithm for the Button problem
has a competitive ratio smaller than $e$.
The statement of the theorem is stronger than this, since the list $b$ does not depend on the algorithm $\cA$ and we will use this fact.
As an aside, the list $b$ is given by $b_i = e^{(i-1)/\Delta}$ (or more precisely, very close approximations by rational numbers),
for a large $\Delta$ that depends on $\epsilon$. See \cite{SLLA25} for details.

An $e$-competitive randomized algorithm for the Button  problem is implicit in  \cite{SLLA25} (and an $(e+\eps)$-competitive algorithm follows 
explicitly from their Theorem 4.1 and Lemma 5.1),
and also it is straightforward to verify, as in the computation from Section \ref{s_alg}, that the following works:
choose uniformly at random a $y \in [0,1)$, and
in iteration $i$, press the  button $j := \arg\max\{h : b_h\le e^{i-1+y}\}$.

\begin{proof}[\emph{of Theorem \ref{thm: lb_nested_randomized}}]
    We use exactly the same construction as in the proof of Theorem \ref{thm: lb_nested_deterministic},
    except that the vector $b$ and the $j^*_q$'s are chosen such that  $\mathbb{E}[\cost(\cA_q,b,j^*_q)] > (e - \eps/4) \cdot b_{j^*_q}$,
    where $\cA_q$ is a randomized online algorithm for the Button problem that we describe next.
    At the moment of mega-release $q$, the randomized algorithm $\cA$ can be in a number of configurations,
    each with a certain probability associated with it. The oblivious (but otherwise all-powerful)
    adversary can  compute these probabilities, and then obtain a randomized algorithm  for Span Minimization $\cA'_q$
    that follows how $\cA$ handles the jobs of mega-release $q$ assuming no further job releases. From this $\cA'_q$,
    the adversary can deduce $\cA_q$ as in the previous proof.
    As a result, $\cA_q$ is a randomized online algorithm for the Button problem that is known to the
    adversary based on the adversary's knowledge of $\cA$ and of mega-releases $1$ through $q$,
    and so the adversary can indeed choose $j^*_q$.

    The $j^*_q$'s are chosen from $\cA$'s \emph{distribution}, not its realized random choices, 
    so the whole instance is fixed before $\cA$ runs. 
    The rest of the proof is exactly the same as in the previous proof,
    but we apply Theorem \ref{thm: lb_button_e} to see that the expected  cost of  $\cA$ in phase $q$ exceeds $(e - \eps/4 ) \cdot b_{j^*_q}$,
    and then $e$ replaces $4$ in the
    following inequalities, with the algebraic manipulations holding for $e$ as well.
\end{proof}

\section{Agreeable Uniform Processing Time Jobs - Lower Bound}
\label{sec: adversary-2}

Here, we prove a lower bound of $2$ for the competitive ratio of any deterministic online algorithm.
Our construction results in an agreeable instance of jobs all with processing time $1$,
which implies that the deterministic online
algorithm  from Section 2.2 of \cite{albers2025onlinebusytimescheduling} (v2) is optimal.
Recall that for Theorem \ref{t_2}, the online algorithm is not allowed to abort and restart jobs.

\medskip

\begin{proof}[\emph{of Theorem \ref{t_2}}]
Let $\eps > 0$ and assume, for a contradiction, that a
deterministic online algorithm $\cA$ 
has competitive ratio at most $2-\epsilon$. 

We can assume without loss of generality that $\epsilon \leq 1$.
Let $M \geq \lceil 10 / \epsilon \rceil$ be an even integer. 
All the jobs will have processing time $1$.
The adversary releases a first {\em lead} job at time 0, with deadline
$M^2$, and then waits until this job is scheduled by the algorithm,
say, to start at time $t(1) \leq M^2 - 1$.
From the moment $t(1) + 1/M$ on, the adversary starts releasing jobs, all with deadline
$2 M^2$, at time $1/M$ plus the starting time (as decided by the algorithm) of the
previous job.  It will continue to do so as long as the algorithm
starts the newest job so that its execution interval intersects 
that of the preceding job.
Jobs scheduled this way are called {\em follower} jobs.
At some moment, the algorithm $\cA$ has to stop intersecting
the latest job with the previous one, or it would incur a cost
of at least $M^2$ while a solution of cost 2 exists:
all the jobs except the lead one can be scheduled to start at $2 M^2 - 1$.
See Figure \ref{f_lb_2} for an illustration.

The job that the algorithm does not start intersecting the previous job
becomes the second {\em lead job}. The adversary does nothing until time $t(2)$,
when  the algorithm $\cA$ schedules the start of the second lead job.
From the moment $t(2) + 1/M$ on, the adversary starts releasing jobs, all with deadline
$3 M^2$, at time $1/M$ plus the starting time (as decided by the algorithm) of the
previous job.  It will continue to do so as in the case of
the first lead job.
Once again, the algorithm must stop scheduling the
latest job intersecting the previous one (since it would incur a cost of at least $M^2$ while a solution of cost 3 exists).
This way we obtain a third lead job, and continue to the third round.

Continuing this way, 
if the algorithm only has $h \leq M$ lead jobs,  we obtain a contradiction, as follows.
The $h^{th}$ lead job is scheduled by $\cA$ to start at time $t(h)$, where $t(h) \leq h M^2 -1 $.
If there is no other lead job, then the follower jobs that the adversary releases,
all with deadline $(h+1)M^2$, are scheduled
in an interval $I$ starting at $t(h)$ and ending  at $(h+1)M^2$.
Then the length of $I$ is at least $1+ M^2$, while a schedule of cost $h+1$ exists:
each job is scheduled at its starting deadline. The competitive ratio of $\cA$ would be at least $(1+M^2)/(1+M) \geq 2$ 
(the inequality holds since $M \geq 10$).

The adversary 
continues for $M$ rounds, with $M+1$ lead jobs, and stops after $\cA$ schedules the $(M+1)^{th}$ lead job.
Notice that all the jobs released by the adversary are agreeable.

Let $\alpha_i$ be the length of the interval used by the $i^{th}$ lead job and its
followers.  Note that such intervals are  disjoint, from the definition of lead jobs.
The algorithm's cost is $1 + \sum_{i=1}^M \alpha_i$
(the 1 is for the $(M+1)^{th}$ lead job that does not have followers).

We consider two other possible solutions to this instance (see Solution 1 and Solution 2 in Figure \ref{f_lb_2}).
One possible solution is to schedule the first lead job with all its followers and the second lead job in an interval of length
at most $\alpha_1 + 1/M$.
The follower jobs of the second lead job are scheduled with the third
lead job (they all have the same deadline), and the followers of the
third job, and the fourth lead job, in an interval of length at most
$\alpha_3 + 1/M$. This continues, and we obtain a cost of at most
$1 + \sum_{i=1}^{M/2} ( \alpha_{2i-1}  + 1/M)$.

Another possible solution is to schedule the first lead job as is,
and schedule all its followers to overlap the schedule of the
second lead job. All the followers of the second lead job, and
the third lead job can be scheduled in an interval of length at most
$\alpha_2 + 1/M$. The followers of the third lead job are scheduled
to overlap the fourth lead job, and all the followers of the fourth
lead job, and the fifth lead job can be scheduled in an interval of length
at most $\alpha_4 + 1/M$.
This continues, and we obtain a cost of at most
$1 + \sum_{i=1}^{M/2} ( \alpha_{2i}  + 1/M)$.

\begin{figure}[th]
\begin{center}\leavevmode%
\scalebox{0.88}{
  \includegraphics{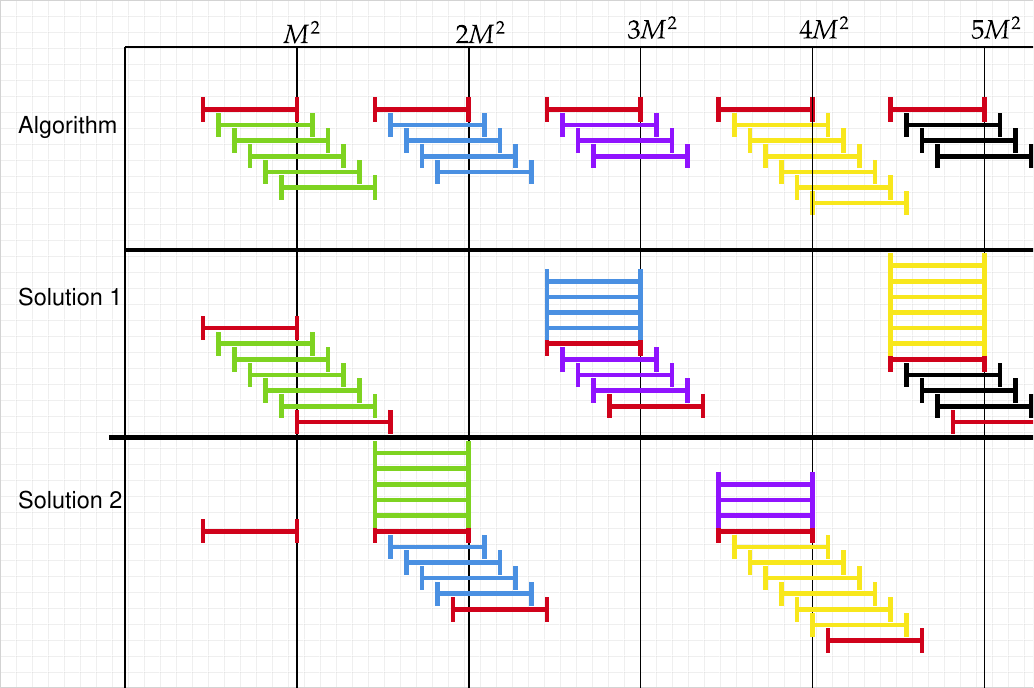}
}
\end{center}
    \caption{An illustration for the proof of Theorem \ref{t_2}. The lead jobs are in red.
    The followers of lead job $1$ are in green, and have the same deadline as lead job $2$.
    The followers of lead job $2$ are in blue, and have the same deadline as lead job $3$.
    The followers of lead job $3$ are in purple, and have the same deadline as lead job $4$.
    The followers of lead job $4$ are in yellow, and have the same deadline as lead job $5$.
    The followers of lead job $5$ are in black.
    }
    \label{f_lb_2}
\end{figure}

As $\sum_{i=1}^M \alpha_i \geq M$ (since each $\alpha_i \geq 1$), we obtain that
the sum of the costs for these two possible solutions is at most
$(1 + 3/M)$ times the cost of the online algorithm.
This then implies that the solution with the lower cost of the two is at most $(1 + 3/M)/2$ times the cost of the online algorithm.
With $\opt$ denoting the cost of an optimum solution and $\cost(\cA)$  denoting the cost of the output of algorithm $\cA$,
we obtain: 
\begin{equation*}
    \opt \leq \frac{1 + 3/M}{2} \cost(\cA) \leq
    \frac{1 + 3/M}{2}  (2 - \epsilon) \opt,
\end{equation*}
from which we deduce  $\frac{1 + 3/M}{2}  (2 - \epsilon) \geq 1$, which implies $M < 6/\epsilon$, a contradiction.
\end{proof}

Notice that this proof fails when $\cA$ is allowed to abort (and restart) jobs, as $\cA$ could abort
all the jobs started before, say, time $t(1) + (1/2)$ once it notices that all these ``follower" jobs have the
same deadline. In fact, \cite{LKT26}  obtain a competitive ratio better than $2$ when restarts are allowed.

\paragraph{Remark (flexible-rigid setting with restarts).}

In the flexible-rigid setting where a job's deadline is only revealed at its starting deadline, restarts do not help, as claimed in Theorem \ref{t_fs}.
The method for proving Theorem \ref{t_fs} is the same as in the proof of Theorem
\ref{t_2} but there are subtle differences for the two models, which is
why below we give a proof, trying not to repeat too many arguments.

\begin{proof}[\emph{of Theorem \ref{t_fs}}]
Let $\eps > 0$ and assume, for a contradiction, that a
deterministic online algorithm $\cA$
has competitive ratio at most $2-\epsilon$.
We can assume without loss of generality that $\epsilon \leq 1$.
Let $M \geq \lceil 100 / \epsilon \rceil$ be an even integer.
All the jobs will have processing time $1$.

Throughout, since $\cA$ may restart a job, \emph{the interval in which a job is
executed} means the length-1 subinterval of $\cA$'s schedule in which the job is
\emph{completed}. 
Two jobs are \emph{amalgamated} if the intervals in which they are executed intersect.
Below, the adversary releases some jobs without specifying their deadline.
Later on, the adversary sets deadlines for these jobs 
(that are at least $1$ plus the current time). 
These deadlines are  revealed to the online 
algorithm  at each job's starting deadline.
Recall that the algorithm can make only finitely many restarts.

In the first round, the lead job $1$ has deadline $M^2$.
There will be at most $M$ lead jobs, and
later on, lead job $i$ will have its deadline set to $i M^2$.
We use $j^i_0$ to denote lead job $i$. 
Every deadline assigned by the adversary is of the form $iM^2$ for some $1\le i\le M+1$.

Round $1$ begins at time $t(1)$, 
 when $\cA$ starts processing lead job $1$.
For $ 1< i<M$, round $i$ begins at time $t(i)$, 
the moment after the end of round $i-1$
 when $\cA$ starts processing lead job $i$. In round $i<M$, the
adversary releases one job every $1/M$ time units, at times
$t(i)+1/M,\ t(i)+2/M,\dots$. These jobs, not counting lead job $i$, form
\emph{train $i$}, ordered by release time. 
The adversary  stops releasing jobs in one of the following two situations: 
\begin{enumerate}
    \item $\cA$ aborts lead job $i$
    \item a job of the train is not amalgamated with any of its strict predecessors, where the predecessor of the first job in train $i$ is lead job $i$ 
\end{enumerate}

In the first situation, the adversary sets the deadline of all the jobs in the train $i$ to be 
$(i+1) M^2$ (this deadline will be revealed to the algorithm at the jobs' starting deadline, which we have not reached yet as the deadline of 
lead job $i$ is $i M^2$). 
When $\cA$ starts processing lead job $i$ again,
we create another train, which with abuse of notation we  still call train $i$.
For the purpose of  Solution 1 and Solution 2, 
the jobs in the old train $i$ are added to a set $X_i$ 
(initially empty).
Moreover, as $\cA$ must complete lead job $i$ by time $i M^2$,
it can abort only finitely often and the final attempt must run to completion.

As before, we can assume that the point where a job of the train is not 
amalgamated with any of its predecessors is always reached:
if train $i$ never ``breaks", $\cA$ is busy on an interval of length at least $M^2/4$, while
scheduling every job at its starting deadline costs at most $i+2$, and as in the proof of Theorem \ref{t_2},
we obtain a contradiction to the fact that the competitive ratio is less than 2.

Let the jobs of train $i$ be $j^i_1, j^i_2, \ldots$, sorted by release time.
Let $z$ be the time when the execution of lead job $i$ is finished.
In the second situation, where lead job $i$ finishes executing, there is 
a job of the train that
is not amalgamated with any of its predecessors. 
Consider the first (by the release time) such job $j^i_{k(i)}$ 
for which this happens. We must have a time $t$ such that all of the predecessors of $j^i_{k(i)}$ have been completed by time $t$ and  $j^i_{k(i)}$
has not been completed and is not being processed  or has  just been aborted. 
Take $t^i_*$ to be smallest $t > z$ such that at time $t$ there is a job $j^i_h$ 
of the train $i$ such that
all of the predecessors of $j^i_{h}$ have been completed by time $t$
and  $j^i_{h}$ has not been completed
and is not being processed at time $t$ or has just been aborted. 
It is easy to check that $h=k(i)$.

As before, we can assume that  $ t^i_* -z  < (1/2) M^2$
and since $z \leq i M^2$, we have $ t^i_* < (i+1/2) M^2$.
At time $t^i_*$, the adversary can detect this situation,
and ends round $i$. The adversary stops releasing jobs
and $j^i_{k(i)}\text{ becomes lead job }i+1.$
Also,
the adversary sets the deadlines of $j^i_1, \ldots, j^i_{k(i)}$  to be $(i+1) M^2$.
The adversary sets the deadlines of $j^i_{k(i)+1}, j^i_{k(i)+2},  \ldots $  to be $(i+2) M^2$.
The adversary then waits until $\cA$ starts or restarts
processing lead job  $i+1$ at time $t(i+1)$, when round $i+1$ begins.

The process ends after lead job $M$ is designated by the adversary,
and no further jobs are released.
One can easily check that deadlines are nondecreasing in release time,
so the instance is again agreeable.

For $1 \leq h \leq M-1$, let $I'_h$ be the interval that the machine is on from the last and 
final start of lead job $h$ to the completion time of all the jobs in 
$\{ j^h_0, j^h_1, \ldots, j^h_{k(h) -1} \}$. 
From the discussion above, $I'_h$ is indeed an interval.
 Let $\alpha_h = \lambda(I'_h)$.
Moreover, the intervals $I'_h$ and $I'_{h'}$ are disjoint for $h \neq h'$.
The cost of $\cA$ is at least $1 + \sum_{i=1}^{M-1}\alpha_i$. Also note that $\alpha_i \geq 1$.
For $1 \leq h \leq M-1$, let $\bI_h$ be obtained from $I'_h$ by increasing the right endpoint by $1/M$.
Let $\hI_1$ and $\hI_M$ be the length-1 intervals where lead job 1 and lead job $M$
are executed respectively.

We consider two other possible solutions to this instance,
Solution 1 and Solution 2.
Solution 1 opens all the intervals $\bI_h$ for  odd $h$, and $\hI_M$.
Solution 1 has a cost of at most
$1 + \sum_{i=1}^{M/2} ( \alpha_{2i-1}  + 1/M)$.
Solution 2 opens all the intervals $\bI_h$ for  even $h < M$, and $\hI_1$ and $\hI_M$.
Solution 2 has a cost of at most
$2 + \sum_{i=1}^{M/2 - 1} ( \alpha_{2i}  + 1/M)$.

Let us verify that Solution 1 is feasible.
For $h$ odd:
\begin{itemize}
    \item 
all the jobs of 
$\{ j^h_0, j^h_1, \ldots, j^h_{k(h)} \}$ (this includes lead job $h$ and lead 
job $h+1$) can be scheduled in $\bI_h$.
\item
all the jobs of $X_h$  can be scheduled in $\bI_h$.
\item 
all the jobs of 
$\{ j^h_{k(h)+1}, j^h_{k(h)+2} \ldots  \}$ are scheduled in $\bI_{h+2}$ except when $h=M-1$,
when all the jobs of 
$\{ j^{M-1}_{k(M-1)+1}, j^{M-1}_{k(M-1)+2}, \ldots  \}$ are scheduled in $\hI_{M}$.
\end{itemize}
For $h < M$ even, all the jobs of $X_h$ and of
$\{ j^h_1, \ldots, j^h_{k(h)}, j^h_{k(h)+1}, \ldots \}$ (this includes lead job $h+1$ and excludes lead job $h$)
are scheduled in $\bI_{h+1}$. 
The lead job $M$ 
is scheduled in $\hI_M$.
Thus Solution 1 is feasible.

Let us verify that Solution 2 is feasible.
For $h < M$ even:
\begin{itemize}
    \item 
all the jobs of 
$\{ j^h_0, j^h_1, \ldots, j^h_{k(h)} \}$ (this includes lead job $h$ and lead 
job $h+1$) can be scheduled in $\bI_h$, 
    \item 
the jobs of $X_h$  can be scheduled in $\bI_h$.
\item
all the jobs of 
$\{ j^h_{k(h)+1}, j^h_{k(h)+2} \ldots  \}$ are scheduled in $\bI_{h+2}$ except when $h=M-2$,
when all the jobs of 
$\{ j^{M-2}_{k(M-2)+1}, j^{M-2}_{k(M-2)+2}, \ldots  \}$ are scheduled in $\hI_{M}$.
\end{itemize}
The lead job $M$ can be scheduled in $\hI_M$ and
the lead job $1$ can be scheduled in $\hI_1$.

For $h$ odd, all the jobs of $X_h$ and of
$\{ j^h_1, \ldots, j^h_{k(h)}, j^h_{k(h)+1}, \ldots \}$ (this includes lead job $h+1$ and excludes lead job $h$) are scheduled in $\bI_{h+1}$
except for $h=M-1$  when these jobs are scheduled in $\hI_M$.
Thus Solution 2 is feasible.

The arithmetic works out similarly to the previous proof,
and we conclude that no deterministic online algorithm can be 
$(2-\epsilon)$-competitive in the flexible-rigid setting, even with restarts
and agreeable jobs.

\end{proof}

The previous discussion on the flexible-rigid scenario resembles the proof of Theorem 1 of
version v2 of \cite{albers2025onlinebusytimescheduling} and maybe they meant to deal with the flexible-rigid case and restarts,
but they never state that this is the case. We did mention in the introduction that we have serious doubts that their Theorem 1 achieves what
our Theorem \ref{t_2} does.

\section{Conclusions}
\label{s_concl}
We have obtained lower and upper bounds of $e$ and $1+e$ respectively
on the competitive ratio of randomized online algorithms for Span Minimization.  
The randomized algorithm presented  has a better competitive
ratio than any deterministic algorithm can achieve.

As in the deterministic case, there is a gap of $1$ between the lower and upper bounds.
Reducing this gap, for both deterministic and randomized algorithms, is of interest.
Our intuition is that none of these four bounds is tight, as the interplay between nested, w-disjoint, and agreeable pairs of jobs
makes online Span Minimization seemingly harder than the much simpler Button problem. If we were to suggest the easiest to improve,
it would be the deterministic upper bound of 5, with the help of restarts and lookahead.

The example(s) from Figure \ref{f_example}   do not show that 
the competitive ratio of $(1+e)$ is tight for our randomized algorithm. 
This is because, for many flag jobs $j$, only part of  $RI(j)$ is actually used.
Appendix \ref{a_tight} proves that 
the competitive ratio of $(1+e)$ is tight for our randomized algorithm. 

Our results and techniques may lead to improved
randomized algorithms for some bounded capacity scenarios, similar to Section \ref{s_cap}.
A further generalization allows heterogeneous machines,  where several machine types, of different costs and
capacities, are available, and in the objective function
the machine-dependent cost is multiplied by the machine's busy time.
With uniform processing times, \cite{CDKZ24} provides constant competitive ratios.

With heterogeneous machines and inflexible jobs of different processing times,
\cite{Ren-tang} and \cite{liu2021analysis} 
obtain  $O(1)$-approximation algorithms (with constants 9 and 14, respectively) in the offline setting
and  a $\Theta(\mu_w)$-competitive algorithm in the online setting, for 
$ \mu_w=\frac{\max_j(d_j-r_j)}{\min_j(d_j-r_j)} $.
In these works, jobs have  \emph{heights}, where each job may take up non-unit space on each machine.
\cite{liu2021analysis} also provide matching lower bounds for both settings.

Another direction for future research would be a good model for learning-augmented algorithms.
\cite{LT24} and  \cite{NEURIPS2021_8b838818} discuss possible directions for related problems.
One version where learning-augmented algorithms may give better bounds is the uniform processing times setting with restarts in the
flexible-rigid scenario.

\section*{Acknowledgements} We thank an anonymous WAOA 2026 reviewer, whose comments led to  a minor correction.

\bibliographystyle{alpha}
\bibliography{bib-2026}

\appendix

\section{The $2$-Track OnlineGreedyTracking algorithm fails the inclusion property}
\label{a_albers}

We start by discussing the incorrect result of 
\cite{albers2025onlinebusytimescheduling} (v2).
They present the OnlineGreedyTracking algorithm, with input a set of {\bf rigid} jobs.
A {\em track}, following \cite{CKM17}, is a set of pairwise w-disjoint jobs.
A {\em $2$-Track} consists of two tracks.  Let
$Q_1, Q_2, \ldots$ denote the sequence of $2$-Tracks that the OnlineGreedyTracking algorithm produces, where the ordering is important.
A $2$-Track $Q_i$ consists of tracks $T^1_i$ and $T^2_i$.
A {\em bundle} in the sense of \cite{albers2025onlinebusytimescheduling} (v2)
consists of $\lfloor g/2 \rfloor$  $2$-Tracks.
With $J$ being rigid, $I_j$ is fixed from the input, and we define $\Span(A) := \bigcup_{j \in A} I_j$.
\cite{albers2025onlinebusytimescheduling}  (v2) use
a machine to schedule all the jobs $J_B$ that belong to a bundle $B$, with a contribution to the
objective function of $\lambda(\Span(J_B))$. 

It is crucial that the $2$-Tracks produced by the online algorithm satisfy the {\em inclusion} property:
\begin{equation*}
    \Span\left( T^1_i \bigcup T^2_i \right) \subseteq  \Span\left( T^1_{i-1} \bigcup T^2_{i-1} \right)
\end{equation*}
for all $i > 1$.
Assuming this inclusion property, our analysis is as follows. 
\cite{albers2025onlinebusytimescheduling}'s algorithm, slightly adapted, consists of 
running the algorithm from Theorem \ref{thm: randomized}
on the flexible jobs to fix intervals, followed by running OnlineGreedyTracking  on the resulting rigid instance 
(we believe that this can indeed be done with a $p_{max}$-lookahead, similar to the discussion in Section 3.2 of the v2 version of \cite{albers2025onlinebusytimescheduling}).
From now on we adapt ideas from the proof of Theorem 5 of \cite{CKM17}.
If $g \leq 6$, we get a $6$-competitive algorithm by just placing every job on a machine on its own, from Lower Bound (\ref{e_v}).
If $g > 6$, from Theorem \ref{thm: randomized} we have  
\begin{equation}
    \label{e_B1}
    \mathbb{E} \big[ \lambda(\Span(J_{B_1})) \big] \leq (1 + e) \opt.
\end{equation}
For $i > 1$, we have:
\begin{align*}
    \lambda\left( \Span\left(J_{B_i}\right) \right) &= \lambda(\Span(Q_{1 + (i-1) \lfloor g/2 \rfloor}))  \\
    &\leq  \frac{1}{\lfloor g/2 \rfloor } \sum_{k=1}^{\lfloor g/2 \rfloor} \lambda(\Span(Q_{k + (i-2) \lfloor g/2 \rfloor}))\\
    &\leq  \frac{1}{\lfloor g/2 \rfloor } \sum_{k=1}^{\lfloor g/2 \rfloor} 
    v(T^1_{k + (i-2) \lfloor g/2 \rfloor}) +  v(T^2_{k + (i-2) \lfloor g/2 \rfloor})\\
    &=  \frac{1}{\lfloor g/2 \rfloor } v(J_{B_{i-1}}) \\
    &= \frac{g}{\lfloor g/2 \rfloor} \frac1g v(J_{B_{i-1}}) \\
    &\leq \frac73 \frac1g v(J_{B_{i-1}}),
\end{align*}
where the last inequality uses $g > 6$.
Adding these inequalities over $i$ and using Inequality (\ref{e_B1}) we obtain that the expected cost of the output is
\begin{equation}
    \label{e_ratio}
    \mathbb{E} \big[ \sum_{i} \lambda\left( \Span(J_{B_i}) \right)  \big] \leq (1+e) \opt + \frac73 \frac1g v(J) \leq (\frac{10}{3} +e) \opt_c \leq 6.06 \; \opt_c,
\end{equation}
where we used Lower Bound (\ref{e_v}) for the last inequality. The worst case of this analysis  is when $g=7$
and one gets a randomized $\left( 1+e + \frac{g}{\lfloor g/2 \rfloor} \right)$-competitive ratio for all $g>1$,
assuming a lookahead of $p_{max}$. Inequality (\ref{e_ratio}) is further improvable
using the $7^{th}$ track as in Subsection \ref{ss_3tracks}.
In the deterministic setting, this analysis (with $g\leq 7$ a separate case) gives a competitive ratio of at most $7.25$.

Unfortunately, the OnlineGreedyTracking algorithm of \cite{albers2025onlinebusytimescheduling} (v2) fails to get the 
inclusion property. Depending on how their ``reservations"
are interpreted (whether a single job or an entire chain is reserved),
one of the following two instances causes the inclusion property to fail.
\begin{enumerate}
    \item six rigid (so $d_j = r_j + p_j$) jobs:
    $r_1 = 0, p_1 = 99, r_2 = 10, p_2 = 99, r_3 = 20, p_3 = 99, r_4=80, p_4 = 99, r_5 = 100, p_5 = 99, r_6 =178, p_6 = 30$.
    \item seven rigid jobs:
    $r_1 = 0, p_1 = 99, r_2 = 10, p_2 = 99, r_3 = 20, p_3 = 99, r_4=80, p_4 = 99, r_5 = 100, p_5 = 99, r_6 = 110, p_6 = 99, r_7 = 198$, $p_7 = 20$.
\end{enumerate}
Looking at other possible algorithms, putting job 2 on $2$-Track 1 also fails, on the second instance above.
As an aside, we cannot think of a reasonable algorithm that would put job 3 on $2$-Track 1.
We tried other variants for OnlineGreedyTracking and failed, and
it may be the case that the inclusion property cannot be satisfied for $2$-Tracks with only $p_{max}$-lookahead.
Instead, we can use $3$-Tracks rather than $2$-Tracks and modify Algorithm~4 of
\cite{albers2025onlinebusytimescheduling},
and then we can prove the inclusion property, in Subsection \ref{ss_3tracks}.

\section{Tightness for the algorithm from Theorem \ref{thm: randomized}}
\label{a_tight}

We next show that the analysis of Theorem~1 is tight for the
standard implementation, whose cost is the union of the jobs' actual
execution intervals rather than the union of the activated flag
intervals, as in the auxiliary variant from Figure \ref{f_example}.
The proof of the next theorem combines the examples from
Figure \ref{f_example} with the examples from Figure \ref{fig: lb4}.

\begin{theorem}
\label{thm:tightness-standard}
For every $\epsilon>0$, there is a finite instance $J$ such that the
randomized algorithm of Theorem~1 satisfies
\[
  \frac{\mathbb{E}[\mathrm{ALG}(J)]}{\opt(J)}
  \ge 1+e-\epsilon.
\]
The competitive ratio of that algorithm is exactly $1+e$.
\end{theorem}

\begin{proof}
We give a family of instances depending on three (very large) integers
$\Delta,H,N$, which will be chosen at the end of the proof.  We assume
$\Delta\ge 2$ and $H\ge 1$.  Set
\[
  m:=2H\Delta+1,\qquad
  k:=H\Delta+1=\left\lceil\frac m2\right\rceil,
  \qquad
  b_i:=e^{(i-1)/\Delta}\quad(1\le i\le m).
\]
Write
\[
  \theta:=b_k=e^H .
\]
Notice that $b_m=e^{2H}=\theta^2$.  Let
\[
  \beta':=m+2,\qquad
  S_q:=(m+2)^{N-q}\quad(1\le q\le N+1),
\]
so that $S_{N+1}=1/(m+2)$.  Starting with $L_1:=0$, define
\begin{equation}
  L_{q+1}
  :=
  L_q+\beta' b_m\bigl(kS_q+S_{q+1}\bigr),
  \qquad 1\le q\le N.
  \label{eq:tightness-phase-starts}
\end{equation}

At time $L_q$, for each $i\in\{1,\ldots,m\}$, release a flexible job
$j_i^q$ with
\begin{equation}
  p_i^q=b_i,
  \qquad
  d_i^q=L_q+\beta' b_m i S_q.
  \label{eq:tightness-megarelease}
\end{equation}
We call these jobs the $q$th mega-release.

Define Phase $q$ to be the interval $[L_q,L_{q+1})$.
We also add rigid jobs that fill an interval of length $\theta$ at the
beginning of every phase.  Choose an integer $\nu$ so large that, with
$\rho:=\theta/\nu$, we have $\rho<e^{-2}$.  For
$a=0,\ldots,\nu-1$, add a rigid job $\pi_a^q$ with
\begin{equation}
  r(\pi_a^q)=L_q+a\rho,\qquad
  p(\pi_a^q)=\rho,\qquad
  d(\pi_a^q)=L_q+(a+1)\rho.
  \label{eq:tightness-rigid-jobs}
\end{equation}
Thus the rigid jobs of phase $q$ have union
\[
  G_q:=[L_q,L_q+\theta).
\]
Note that $|J| = N(m+\nu)$.
We first describe a feasible offline schedule.  Since
\[
  d_1^q-L_q=\beta' b_mS_q>\theta,
\]
the interval $G_q$ lies in the window of every $j_i^q$. 
Since $p_i^q \leq \theta$ for $i \leq k$,
the offline schedule can execute all the jobs $j_i^q$ with $i \leq k$
in $G_q$, in parallel with the rigid jobs.

It is easy to check that 
all remaining jobs can be put into one common final interval
$I^\star:=[L_{N+1},L_{N+1}+b_m)$. Thus:
\begin{equation}
  \opt(J) \leq N\theta+b_m=N \theta +\theta^2.
  \label{eq:tightness-opt}
\end{equation}

We now lower-bound the cost of the online algorithm.
Recall the random choice $y\in[0,1)$.
Every flag interval opened by a rigid job has length
strictly smaller than $e\rho<1$.  Since every flexible job has
processing time at least $b_1=1$, no flexible job is scheduled in such
an interval.  The rigid jobs force actual span exactly
$\theta$ inside every $G_q$.

We next show that the relevant flag intervals of different phases,
and the relevant flag intervals within one phase, are disjoint.  
Indeed, let
\[
  \bs_i^q:=d_i^q-b_i
\]
be the starting deadline of flexible job $j_i^q$.  The first such starting deadline
is strictly later than the end of $G_q$ and every flag interval opened
in $G_q$ by a rigid job.  
 Therefore no flexible job is executed in  any of these  $G_q$. Moreover,
\begin{equation}
  \bs_{i+1}^q-\bs_i^q = \beta' b_mS_q-(b_{i+1}-b_i) > (\beta'-1)b_m.
  \label{eq:tightness-start-gap}
\end{equation}
Every flag interval used before the first $k$ flexible jobs have been scheduled
has length less than $e \theta\le b_m$.  Also,
\begin{equation}
  L_{q+1}-\bs_k^q = \beta' b_mS_{q+1}+\theta > e\theta.
  \label{eq:tightness-phase-gap}
\end{equation}
Thus all the flag intervals that we count in phase $q$ are mutually
disjoint  and end before $L_{q+1}$.

Every job left over after the first $k$ flexible jobs of a mega-release
have been covered has starting deadline after all $N$ phases.
If such a job $j$ becomes a flag job,
no job $j'$ from the first $k$ flexible jobs of all the mega-releases 
can be scheduled in $RI(j)$ since the deadline of $j'$ is in one of  the phases.
We ignore the contribution of jobs scheduled after all the $N$ phases.

Fix a phase $q$.  Suppose that $j_i^q$ opens a flag interval of length
$\ell$, and let $h\ge i$ satisfy $b_h\le \ell$.  Then one can easily check
that $j_h^q$ fits inside this flag interval.  

We claim that, before all the first $k$ jobs of this mega-release have
been scheduled, the algorithm opens flag intervals of lengths
\begin{equation}
  e^y,e^{1+y},\ldots,e^{H+y}.
  \label{eq:tightness-flag-sequence}
\end{equation}
The first unscheduled job has processing time $b_1=1$, so it opens the
interval of length $e^y$.  Inductively, after an interval of length
$e^{t+y}$ has opened, every job of the current mega-release with
$b_i\le e^{t+y}$ is scheduled.  If $t<H$, the smallest processing time
of a remaining job is at most
\[
  e^{t+y+1/\Delta}<e^{t+1+y},
\]
and is larger than $e^{t+y}$; therefore its rounded flag length is
$e^{t+1+y}$.  Since $b_k=e^H$, the interval of length $e^{H+y}$
schedules the first $k$ jobs. 
No other flag event can interrupt this sequence.

For $t=0,\ldots,H$, let
\[
  a_t:=\lfloor\Delta(t+y)\rfloor+1.
\]
Then
\begin{equation}
  e^{-1/\Delta}e^{t+y}
  <
  b_{a_t}
  \le
  e^{t+y}.
  \label{eq:tightness-utilization}
\end{equation}
The choice of $m$ guarantees $a_t\le m$.  For $t\ge1$,
$b_{a_t}>e^{t-1+y}$, so $j_{a_t}^q$ was not scheduled in an earlier
flag interval.  It is therefore scheduled in the interval of length
$e^{t+y}$.  Its actual execution interval has length $b_{a_t}$ and is
contained in that flag interval.  Hence the union of the actual
executions in the flag interval has length at least $b_{a_t}$.

By \eqref{eq:tightness-start-gap}--\eqref{eq:tightness-phase-gap}, the
flag intervals being counted are disjoint.  Thus, for every phase $q$
and every $y\in[0,1)$, the actual span contributed after $G_q$ is at
least
\begin{equation}
  e^{-1/\Delta}\sum_{t=0}^{H}e^{t+y}.
  \label{eq:tightness-phase-cost}
\end{equation}
Taking expectation over the single random choice $y$ gives
\begin{align}
  \mathbb{E}_y\left[\sum_{t=0}^{H}e^{t+y}\right]
  &=
  \left(\int_0^1 e^y\,dy\right)
  \sum_{t=0}^{H}e^t\notag\\
  &=
  (e-1)\frac{e^{H+1}-1}{e-1}
  =
  e\theta-1.
  \label{eq:tightness-rounded-sum}
\end{align}
The same $y$ is used in every phase, but independence is not needed;
linearity of expectation suffices.  Combining the rigid and flexible
contributions,
\begin{equation}
  \mathbb{E}[\mathrm{ALG}(J)]
  \ge
  N\left(\theta+e^{-1/\Delta}(e\theta-1)\right).
  \label{eq:tightness-alg}
\end{equation}

Finally, by \eqref{eq:tightness-opt} and \eqref{eq:tightness-alg},
\begin{align}
  \frac{\mathbb{E}[\mathrm{ALG}(J)]}{\opt(J)}
  &\ge
  \frac{N\left(\theta+e^{-1/\Delta}(e\theta-1)\right)}
       {N\theta+\theta^2}\notag\\
  &=
  \frac{N}{N+\theta}
  \left(
    1+e^{1-1/\Delta}-e^{-H-1/\Delta}
  \right).
  \label{eq:tightness-ratio}
\end{align}
The right-hand side tends to $1+e$ by first taking $\Delta$ and $H$
large and then taking $N/\theta$ large.  Thus it is at least
$1+e-\epsilon$ for very large finite choices of $\Delta,H,N$.

The entire instance is fixed before $y$ is chosen, so the construction
is valid against an oblivious adversary.  If rational input data are
required, all processing times and endpoints may be perturbed to
sufficiently close rational values; 
this changes the final ratio by an arbitrarily small amount.
\end{proof}

\end{document}